\documentclass[12pt,reqno]{amsart}
\usepackage[utf8]{inputenc}

    \usepackage[dvipsnames]{xcolor}
    \usepackage[T1]{fontenc}
	\usepackage[utf8]{inputenc}
	\usepackage[english]{babel}
	\usepackage{amsfonts}
	\usepackage{amssymb}
	    \usepackage{amsmath}
	
	\usepackage{enumerate}
	\usepackage{bbm}
	\usepackage{mathtools}
	\usepackage{subfig}
	\usepackage{booktabs}	
	\usepackage{setspace}
	\usepackage{graphicx}
	\usepackage{fullpage}
	
	\usepackage{soul}
	\usepackage{xspace}
	\usepackage{csquotes}
	\usepackage[margin = 1in]{geometry}
	\usepackage{placeins}
	\usepackage[footfont=normalsize,font=normalsize]{floatrow}
	\usepackage{lscape}
	\usepackage{longtable}
    \usepackage{accents}
    \usepackage{float}

    \usepackage{tikz}
    \usetikzlibrary{calc,intersections}

	\usepackage[longnamesfirst]{natbib}
	\usepackage{bibunits}
	\defaultbibliography{Bibliography.bib}  
	\defaultbibliographystyle{aer}

	\usepackage{hyperref}
	\hypersetup{
    colorlinks=true,
    linkcolor=Red,
    filecolor=magenta,      
    urlcolor=blue,
    citecolor = Blue
}
	\usepackage[capitalise, nameinlink, noabbrev]{cleveref}

    \usepackage{rotating}
	\usepackage[framemethod=tikz]{mdframed}
	\usepackage{todonotes}
	\presetkeys{todonotes}{color=black!5,inline}{} %

    \usepackage{tcolorbox}
    \newtcbox{\feedback}{nobeforeafter,colframe=black,colback=white,boxrule=0.5pt,arc=2pt,
      boxsep=0pt,left=2pt,right=2pt,top=2pt,bottom=2pt,tcbox raise base}
      
    \usepackage{amsthm}
     \newtheoremstyle{theorem} %
    {3pt}                    %
    {3pt}                    %
    {\slshape}                   %
    {}                           %
    {\bfseries}                   %
    {.}                          %
    {.5em}                       %
    {}  %
    \theoremstyle{theorem}

    \newtheorem{prop}{Proposition}[section]

    \newtheorem{lem}{Lemma}[section]

    \theoremstyle{definition}
    \newtheorem{rem}{Remark}
    
    \newtheorem{example}{Example}

    \crefname{prop}{Proposition}{Propositions}
\crefname{assumption}{assumption}{assumptions}
\crefname{asm}{assumption}{assumptions}
\crefname{prop}{proposition}{propositions}
\crefname{cor}{Corollary}{Corollaries}
\crefname{lem}{Lemma}{Lemmas}

\usepackage{array}
\newcolumntype{L}[1]{>{\raggedright\let\newline\\\arraybackslash}m{#1}}
\newcolumntype{C}[1]{>{\centering\let\newline\\\arraybackslash\hspace{0pt}}m{#1}}
\newcolumntype{R}[1]{>{\raggedleft\let\newline\\\arraybackslash\hspace{0pt}}m{#1}}

\usepackage{ragged2e}
\newlength\ubwidth

\newcommand\numberthis{\addtocounter{equation}{1}\tag{\theequation}}

\let\emptyset\varnothing

\newcommand{\indep}{\perp \!\!\! \perp}

	\newcommand{\Var}{\mathop{}\!\textnormal{Var}}

	\newcommand{\parens}[1]{\left(#1\right)}

	\DeclareMathOperator*{\argmax}{arg\,max}

\newgeometry{margin=1.25in} %

\usepackage{macros}
\usepackage{tikz}
\DeclareMathOperator{\dist}{dist} %
\DeclareMathOperator{\Log}{Log} %

\title{Testing Monotonicity in a Finite Population} %
\author{Jiafeng Chen \and Jonathan Roth \and Jann Spiess}
\date{\today. (First public version: December 31, 2025.) Chen: Department of Economics, Stanford University, jiafeng@stanford.edu; Roth: Department of Economics, Brown University, jonathan\_roth@brown.edu; Spiess: Graduate School of Business, Stanford University, jspiess@stanford.edu. We thank Isaiah Andrews, Neil Christy, Dean Eckles, Guido Imbens, Toru Kitagawa, Amanda Kowalski, Lihua Lei, Matt Masten, Ashesh Rambachan, Jesse Shapiro, and Kaspar Wüthrich for insightful discussions. Roth gratefully acknowledges funding from the NIH under grant NIGMS 1R35GM155224 and from the Alfred P.\ Sloan Foundation.}
\begin{document}

\maketitle

\begin{abstract}
We consider the extent to which we can learn from a completely randomized experiment
whether all individuals have treatment effects that are weakly of the same sign, a condition we
call \emph{monotonicity}. From a classical sampling perspective, it is well-known that monotonicity is not falsifiable. We show that from the design-based perspective—in which the potential outcomes are fixed and only treatment assignment is stochastic—that the distribution of treatment effects in the
finite population (and hence whether monotonicity holds) is formally \emph{identified}. Nevertheless, we show that the scope for learning about violations of monotonicity is severely limited. Frequentist tests of monotonicity have generically poor power, and there exist (non-degenerate) Bayesian priors that never update about whether monotonicity holds. Estimators of the magnitude of the violation of monotonicity are likewise shown to have poor minimax mean-squared error.

\end{abstract}

\newpage 

\section{Introduction}

Let $D$ be a randomly assigned binary treatment and $Y$ a binary outcome. Researchers are
often interested in evaluating the \emph{monotonicity} assumption that everyone has a treatment effect weakly of the same sign.%
\footnote{Also known as monotone treatment response \citep{manski1997monotone}.}
For
example, if $D$ is a medical intervention, we may be interested in whether some patients benefit from the treatment while others are harmed, which would indicate that outcomes could potentially be improved by better targeting.
Likewise, if $D$ is an
instrumental variable and $Y$ a treatment with imperfect compliance (more frequently
denoted $Z$ and $D$, respectively), then the monotonicity assumption is needed for
two-stage least squares to have a local-average-treatment-effect interpretation
\citep{AngristImbens(94)}.

There are two distinct approaches to statistical uncertainty in causal inference settings
that lead to different formal definitions of monotonicity. To discuss these, let a \emph{type} denote one of four possible values for potential outcomes $(y(1), y(0))$. Borrowing
from the instrumental-variables literature, we use $at$ to denote an ``always-taker'' with
$y(0)=y(1)=1$, and likewise use $nt$ for never-takers ($y(0)=y(1)=0$), $d$ for defiers ($y(1)=0,y(0)=1$), and $c$ for compliers
($y(1)=1,y(0)=0$). 
The classical perspective views
the $n$ observed units as sampled from an infinite superpopulation. Monotonicity then
imposes that there are no individuals in the superpopulation with opposite-signed
treatment effects. Formally, we assume that the potential outcomes are sampled from a
distribution $g^*_p$, $(y_i(1),y_i(0)) \sim g^*_p$, where $p =
(p_{at},p_{nt},p_{d},p_{c})$ denotes the superpopulation shares of each ``type'', where
each type corresponds to one of four possible values for $(y(1), y(0))$). 
The classical definition of
monotonicity is then $\min(p_d,p_c) =0$.

By contrast, the \emph{design-based} perspective views the $n$ units and their potential
outcomes as {fixed} or conditioned upon, with statistical uncertainty arising only from
the stochastic assignment of treatment. The finite population is characterized by $\theta = (\theta_{at}, \theta_{nt}, \theta_{d}, \theta_{c})$, where $\theta_t$ is the
\emph{count} of how many of the $n$ units are of type $t$. The design-based monotonicity
assumption is then that $\min(\theta_d,\theta_c) = 0$.%
\footnote{This monotonicity assumption is distinct from the one-sided hypothesis that there are no defiers (or equivalently, no negative treatment effects), which has the (generally testable) implication that the average treatment effect is non-negative. \cite{caughey2023randomisation}  employ randomization inference to construct valid test of the null of no negative treatment effects in a design-based model that also captures non-binary outcomes.}
In sum, the classical
monotonicity assumption asks whether there are any units with opposite-signed treatment
effects in the superpopulation from which the sample is drawn, while the design-based
version asks whether there are any opposite-signed treatment effects \emph{among the $n$
units at hand}.

It is well-known that the observable data are always compatible with the null that monotonicity holds in the
superpopulation. Suppose that we have a completely randomized experiment in
which $n_T$ of the $n$ units are randomly assigned to treatment. The researcher then
observes $(Y_T,Y_U) = (\sum_i D_i Y_i, \sum_i (1-D_i) Y_i)$, summarizing how many units in
each treatment group have $Y_i=1$. From the superpopulation perspective,
we have that $(Y_T,Y_U) \sim g_p$ (where $g_p$ is induced by $g_p^*$ and complete random assignment of the treatment). However, it is straightforward to show that
$p$ is not point-identified \citep[see, e.g.,][]{heckman_making_1997, fan_sharp_2010}. In particular, for any $p_1$ that violates monotonicity,
there exists a $p_0$ satisfying monotonicity such that $g_{p_1} = g_{p_0}$. Hence, from
the superpopulation perspective, the monotonicity assumption is not falsifiable.%
\footnote{We note that monotonicity may be testable when combined with other assumptions. For example, there are tests of the joint assumptions of instrument monotonicity, exclusion, and independence \citep[e.g.][]{kitagawa_test_2015}.
In addition, monotonicity has testable implications for the marginal distributions of treated and control observations when the outcome has more than two levels (\citealp{angrist1995two}; see also Appendix 7.2 of \citealp{caughey2023randomisation} for an application in a design-based framework).
}

In contrast, we show that from the design-based perspective, the type counts $\theta$ are in fact identified according to the usual technical definition of identification. Let $f_\theta$ denote the distribution of the observable data $(Y_T,Y_U)$  over repeated assignments of treatment for the fixed finite population. We prove that the mapping $\theta \mapsto f_\theta$ is one-to-one, so $\theta$ is formally identified. At first blush, the typical definition of identification appears to suggest that the data \emph{are} informative about whether the design-based version of monotonicity holds. 

We argue, however, that the usual notion of identification is somewhat unnatural in the
design-based setting: identification asks what can be learned with repeated
realizations of the outcomes under different treatment assignments for the \emph{same}
units. In practice, however, we only observe one realized treatment assignment for each of
these units. It is therefore not clear whether identification in the usual sense implies
that we can meaningfully {learn} about whether monotonicity holds from the data actually observed. To evaluate the extent to which one can feasibly learn about
monotonicity given a single realization of the data in the design-based setting,
we ask two questions: First, to what extent can one construct frequentist tests of the
null of monotonicity? Second, to what extent do Bayesians update about the probability
that monotonicity holds?

From the frequentist perspective, we show that there exist tests of the design-based
monotonicity assumption with non-trivial power against some alternatives, but we
generally expect the power of these tests to be poor. In particular, we show that any test of monotonicity will have trivial or near-trivial power against some alternative: when $n$ is even, any level-$\alpha$ test has power no larger than $\alpha$ against some alternative; when $n$ is odd, worst-case power is bounded above by $\alpha (1 + O(2^{-n}))$. Moreover, we show
that any test that has power against some alternative $\theta_1$ will have poor power at
alternatives ``close'' to $\theta_1$. We formalize this by deriving upper bounds on weighted average power (WAP) of tests in a neighborhood of $\theta_1$, where the weights
are proportional to the frequencies expected under sampling types from a superpopulation (and
thus concentrate around $\theta_1$, in terms of relative frequencies, when $n$ is
reasonably large). We show that regardless of sample size, WAP of a size-$\alpha$
test
is no greater than $2.51 \alpha$; thus, for a 5\% test, WAP is never more than
12.6\%. The
bound can be even tighter given a particular $n$ and a lower bound on the number of
defiers/compliers. For example, if $n=100$ and compliers and defiers are each at least
5\% of the finite population, then WAP is never greater than 5.06\%. We also derive upper bounds on power for unbiased tests (i.e. those with power at least $\alpha$ for all alternatives). This upper bound converges to $\alpha$ as $n$ grows large, indicating that with large $n$ all unbiased tests have nearly trivial power. 

Taken together, our results show that frequentist tests provide little information about monotonicity.
In particular, these tests can only be useful if we are interested in testing one particular alternative and not
nearby ones. For example, with $n=30$ and $n_T=15$, the highest possible power of a 5\% test against any alternative is 31\%, which is achieved by a test targeting the alternative with 18 defiers and 12 compliers. However, this test never rejects when there are 17 defiers and 13 compliers. Our theoretical results show the low power of this test against neighboring alternatives is a generic feature of tests against monotonicity. We suspect that researchers are unlikely to be interested in such ultra-specific alternatives, and therefore these tests will
have relatively little use in practice. 

From the Bayesian perspective, we show that there do exist Bayesians who update about
the probability that the design-based monotonicity assumption holds, but there also exist
Bayesians who never update. The existence of a Bayesian who updates follows immediately
from the fact that $\theta$ is identified. Consider a Bayesian with a two-point prior on
$\theta_0$ satisfying monotonicity and $\theta_1$ violating monotonicity. Since
$f_{\theta_0} \neq f_{\theta_1}$, the Bayesian updates about the probability that the null
is true. However, we show that there exists a non-trivial prior $\pi$ which never
updates on the probability that the null hypothesis holds, i.e. $\pi(\theta \in \Theta_0
\mid (Y_U,Y_T) ) = \pi(\theta \in \Theta_0) \in (0,1)$ ($\pi$-a.s.), where $\Theta_0$ denotes
the set of parameter values satisfying monotonicity.%
\footnote{Specifically, we show such priors exist when $n$ is even. When $n$ is odd, we show that there exist priors that \emph{minimally} update, in the sense that the expected absolute difference between the prior and posterior probabilities for $\theta \in \Theta_0$ is $O(2^{-n})$; see \Cref{prop:bayesupdateodd}.}
Our results thus suggest that the data may be
informative to \emph{some} audience members, who are worried about particular violations
of monotonicity, but there will not be consensus: some audience member is always unmoved by the data on monotonicity. The existence of a non-updating Bayesian also implies that any classifier of whether $\theta \in \Theta_0$ does no better than random guessing at some parameter value. 

Although we primarily focus on learning \emph{whether} monotonicity holds, we show that it is likewise difficult to estimate the \emph{extent} to which it is violated---in turn implying it is difficult to estimate the full vector of types $\theta$. Specifically, we show that the minimax mean-squared error (MSE) for estimating the monotonicity violation $v = \frac{1}{n} \min(\theta_c, \theta_d)$ is close to that of the trivial guess $\hat{v} = \frac{1}{4}$. We also show that the minimax MSE for estimating the full vector $\hat\theta$ is close to that of guessing the midpoint of the sample Fr\'echet--Hoeffding bounds. Interestingly, we find that the maximum-likelihood estimator (MLE) converges in probability to one of the Fr\'echet--Hoeffding \emph{end}points and therefore has asymptotic worst-case MSE that is larger than minimax. Our results further imply that the MLE converges in probability to a limit in which one of the estimated type shares is zero, even when all four types are present in the population. Thus, despite the likelihood not being completely flat in $\theta$, the extent to which we can estimate $\theta$ beyond its marginals is severely limited.

We now briefly provide some intuition for why the finite-population type-counts $\theta$ are formally identified and yet there is still extremely limited scope for learning about violations of monotonicity from the data. Our results establish that when $n$ is large, the distribution of the observed data $(Y_T,Y_U)$ is close to a discrete normal distribution (in the sense of  vanishing Le Cam distance between the two statistical experiments), with only the correlation of the discrete normal distribution depending on the extent to which monotonicity is violated. Thus, while the distribution of the data depends on the full vector $\theta$, and hence it is identified, learning about monotonicity violations is asymptotically equivalent to learning about the correlation of a bivariate discrete normal distribution \emph{from a single draw}. Of course, this asymptotic argument only establishes that it is difficult to learn about monotonicity violations as $n \to \infty$. Many of our results also establish that there is limited scope for learning even with moderate $n$, which cannot be shown using a large-$n$ normal limit alone. Instead, we establish these results by carefully constructing pairs of priors over parameters $\theta$ in the null and alternative space that induce distributions over $Y$ that are similar to i.i.d.\ samples from observationally equivalent super-populations. This allows us to establish that the information about monotonicity in finite samples is limited and vanishes quickly.

\smallskip
\paragraph{\textbf{Related literature.}} Previous work has derived the likelihood of the
observed data under the design-based data-generating process we consider, in which there
is a completely randomized experiment with binary outcomes, and has shown that this can be
used for likelihood-based or Bayesian inference \citep{copas_randomization_1973,
ding_model-free_2019, christy_counting_2025}. \citet{copas_randomization_1973} and
\citet{christy_counting_2025} explicitly note that the likelihood can depend on the number
of defiers $\theta_d$. We add to this literature an explicit analysis of identification,
as well as results quantifying the extent to which an analyst can test for, update about, or estimate monotonicity violations. These results are also relevant for proposals and discussions on decision problems where utility depends on counterfactual outcomes \citep{ben2024policy, strzalecki2026hippocratic,koch2026axiomatic}.

More broadly, an extensive previous literature has considered the different implications
of sampling-based versus model-based approaches to uncertainty for \emph{inference}
\citep[e.g.][]{li_general_2017,abadie_sampling-based_2020}. This paper highlights that
these different approaches can also have different implications for \emph{identification}.
We argue, however, that the classical notion of identification may be misleading in the
design-based setting as a criterion for whether the data is informative about a parameter,
and instead propose to evaluate the informativeness of the data through the properties of
frequentist tests and Bayesian updating. Although we focus on testing monotonicity, these
observations may prove useful in other design-based settings as well. \citet{kline2025finite} also study a design-based setting (although they do not consider monotonicity testing) and likewise find the textbook definition of identification inadequate, although they opt for defining alternative notions of identification rather than quantifying the scope for frequentist testing or Bayesian updating.

\section{Setup}
\label{sec:setup}
Consider a finite population of $n$ individuals subjected to a completely randomized
experiment with $n_T$ treated units, where $0 < n_T < n$. That is:
\begin{enumerate}
    \item Each individual $i$ has
potential outcomes $(y_i(1), y_i(0)) \in \br{0,1}^2$.
\item We observe $(Y_i, D_i)_
{i=1}^n$, where $Y_i = y_i(D_i)$ and $D_i \in \br{0,1}$.
\item Fixing the finite
population, randomness solely arises from
the assignment of $(D_1,\ldots, D_n)$, where \begin{align*}
&P(D_1=d_1,\ldots, D_n=d_n, Y_1=y_1,\ldots, Y_n = y_n) \\&\quad= \binom{n}{n_T}^{-1} \one
\pr{\sum_
{i=1}^n d_i = n_T; y_i = y_i(d_i) \text{ for all $i$}}.
\end{align*}
\end{enumerate}
We refer to those with $y_i(1) > y_i(0)$ as compliers, $y_i(1) =  1 = y_i(0)$ as always-takers, $y_i(1) < y_i(0)$ as defiers, and $y_i(1) = 0 = y_i(0)$ as
 never-takers.
 We refer to these as the \emph{type} of a unit. We observe the number of treated units with $Y_i=1$ and the number of
untreated units with $Y_i=1$:%
\footnote{We note that if the researcher observes individual data, then any procedure $\delta$ that is \emph{anonymous}---i.e., that is invariant to permutations of the units, so that $\delta((Y_1,D_1),\ldots, (Y_n, D_n)) = \delta((Y_{\sigma(1)}, D_{\sigma(1)}),
\ldots, (Y_{\sigma(n)}, D_{\sigma(n)}))$ for all permutations $\sigma: [n] \to [n]$---is simply a function of the counts $(Y_T,Y_U)$. It thus suffices to restrict attention to $(Y_T,Y_U)$ for anonymous procedures.}
\[
Y := (Y_T, Y_U) = \pr{\sum_{i=1}^n D_i Y_i, \sum_{i=1}^n (1-D_i) Y_i}. 
\]
 
 The counts $Y$ depend on the potential outcomes only through
the
corresponding counts of types in the finite population, \[\theta \in \Theta :=   \br{(n_{at}, n_{nt}, n_d, n_c) \in (
\mathbb{N}\cup
\br{0})^4 : n_ {at} + n_{nt} + n_d + n_c
 = n}.
 \]
 In particular,
 \begin{align*}
 Y_T = (\text{\# treated always takers}) + (\text{\# treated compliers}) \\ 
 Y_U = (\text{\# untreated always takers}) + (\text{\# untreated defiers}).
 \end{align*}
 In a completely randomized experiment, given $\theta = (n_
 {at}, n_{nt}, n_d, n_c)$, the number of treated
 units of each type is distributed according to a multivariate hypergeometric
 distribution with parameters $\theta$ and $n_T$. Let $f_\theta(y_T, y_U)$ denote the induced
 probability mass function for $(Y_T, Y_U)$ under $\theta$. 

Sometimes, we think of the finite population as being drawn from an infinite
superpopulation. To that end, let $\mathcal P = \br{(p_{at}, p_{nt}, p_d, p_c) \in [0,1]^4
: p_{at} +  p_{nt} +  p_d +  p_c = 1}$ be the simplex. For an element $p \in \mathcal P$,
let the corresponding distribution of $(Y_T, Y_U)$ when unit types are drawn i.i.d. from
$\mathrm{Multinomial}(n, p)$ be \[
g_p(y_T, y_U) = \sum_{\theta = (n_{at}, n_{nt}, n_d, n_c) \in \Theta} f_\theta (y_T, y_U)
 \frac{n!}{n_{at}! n_{nt}! n_d ! n_c!} p_{at}^{n_{at}} p_{nt}^{n_{nt}} p_d^{n_d} p_c^{n_c}.
\]

Let $\Theta_0 = \br{(n_{at}, n_{nt}, n_d, n_c)  \in \Theta: \min(n_d, n_c) = 0}$ be the
set of type counts that satisfy monotonicity. Let
$\Theta_1 = \Theta \setminus \Theta_0$ be the complement. Similarly, let $\mathcal P_0 =
\br{p \in \mathcal P: \min(p_d, p_c) = 0}$ and $\mathcal P_1 = \mathcal P \setminus
\mathcal P_0$. 

\begin{rem}[Practical relevance of each null]
Whether we prefer to test the superpopulation null that $p \in \mathcal{P}_0$ or the
design-based null that $\theta_0 \in \Theta_0$ will depend on the application. If the $n$
units are patients in a drug trial drawn randomly from a much larger population of
patients with the same condition, then we are likely more interested in whether there are heterogeneous responses to the drug in the superpopulation of patients, and thus $p \in
\mathcal{P}_0$ is more relevant than $\theta \in \Theta_0$. On the other hand, if the $n$
units are the 50 states, it may be unnatural to imagine the states as sampled from an
infinite superpopulation, rendering conceptual issues for the null $p \in \mathcal{P}_0$.
By contrast, testing $\theta \in \Theta_0$ answers the natural question as to whether any
of the 50 states have opposite-signed treatment effects. See \citet{copas_randomization_1973},
\citet{reichardt_justifying_1999}, and \citet{rambachan_design-based_2025}, among others,
for additional discussion of the relevance of design-based vs.\ superpopulation-based
estimands in general.
For our setting specifically, \cite{gelman2025russian} argue that cases in which decision loss depends on counterfactual outcomes should be analyzed through a framework in which potential outcomes are stochastic.

\end{rem}

\section{Identification}

The typical textbook definition of identification states that a parameter is \mbox{(point-)}\allowbreak{}identified if two
distinct values of the parameter induce different distributions of the observed data.%
\footnote{For example, the Wikipedia page on \href{https://en.wikipedia.org/w/index.php?title=Identifiability&oldid=1292274414}{identifiability} states that a statistical model $P_\theta$ is ``identifiable if the mapping $\theta \mapsto P_\theta$ is one-to-one.'' \citet[][Definition 5.2]{lehmann_theory_1998} equivalently define $\theta$ to be unidentifiable if there exist $\theta_1 \neq \theta_2$ such that $P_{\theta_1} = P_{\theta_2}$.}
It
is well-known that, from the superpopulation perspective, the population shares $p$ are not
point-identified \citep{heckman_making_1997}. In particular, any type proportion that
violates monotonicity induces data that can be rationalized by some  type proportion
that obeys monotonicity. By contrast, as summarized in the following result, the type
counts $\theta$ are in fact identified, and thus monotonicity violations are likewise
distinguishable from $f_\theta$.

\begin{restatable}{prop}{propid}
\label{prop:id}
If $\theta_1 \neq \theta_0 \in \Theta$, then $f_{\theta_1} \neq f_{\theta_0}$, and hence the finite-population type counts $\theta$ are identified. On the
other hand, given any $p_1 \in \mathcal P_1$, there exists a $p_0 \in \mathcal P_0$ such
that $g_{p_1} = g_{p_0}$. 
\end{restatable}

\begin{example}[Illustrative example] \label{ex:2 units}
Consider a population with 2 units, one of whom is assigned to treatment. If there is 1 always-taker and 1 never-taker ($\theta = (1,1,0,0)$), then $(Y_T,Y_U)$ is either equal to $(1,0)$ or $(0,1)$, depending on whether the always-taker is assigned to treatment or control. By contrast, if there is 1 defier and 1 complier ($\theta = (0,0,1,1)$), then $(Y_T,Y_U)$ is either equal to $(1,1)$ or $(0,0)$, depending on whether the complier is assigned to treatment or control. Thus, the distribution of the observed data differs between a population with 1 always-taker and 1 never-taker and a population with 1 complier and 1 defier, despite $E[(Y_T,Y_U)]$ being the same. By contrast, a superpopulation with half always-takers and half never-takers and a superpopulation with half compliers and half defiers would each generate the same observable data distribution for $(Y_T,Y_U)$, assigning equal probability to $(0,0), (0,1), (1,0), (1,1)$. 
\end{example}

As a criterion for whether the observed data is informative of the parameter, however, the
above definition of identification falls short from the design-based perspective. Two
values of $\theta$ are distinguishable in the sense that repeated draws of $(Y_T, Y_U)$
from $f_\theta$ would have different distributions under these two values. This, however,
corresponds to knowing the distribution of outcomes from re-assigning the \emph{same}
units to \emph{different} treatment assignments. However, in the finite population, we only
observe $(Y_T, Y_U)$ once, and thus cannot use information only learned from repeated
draws. For example, knowing $f_\theta$ implies that we know $\cov(Y_T,Y_U)$, yet this is
difficult to learn from observing a {single} realization of $f_\theta$.

In the following sections, we consider two perspectives of evaluating whether a single realization of $(Y_T,Y_U)$ is useful for learning about monotonicity violations. From the frequentist perspective, we consider whether one can construct tests that have non-trivial power against the null of monotonicity. From the Bayesian perspective, we consider the extent to which a Bayesian updates their prior that monotonicity holds after seeing the data. 

\begin{rem}
The discussion in previous papers sometimes suggests that $\theta$ is
not identified. For example, \citet{ding_model-free_2019} write that ``Without
monotonicity, the unknown parameters in the Science Table, $(N_{11}, N_{10}, N_{01},
N_{00})$ [$\theta$ in our notation], are no longer identifiable from the observed data.''
Likewise, \citet{rosenbaum_effects_2001} writes, ``The model of a nonnegative effect
cannot be verified or refuted by inspecting the responses of individuals, because $r_{Ti}$
and $r_{Ci}$ [$y_i(1)$ and $y_i(0)$ in our notation] are never jointly observed on the
same person.'' \Cref{prop:id} shows that according to the usual technical definition of
identification, $\theta$ is in fact identified. The underlying intuition---that violations of monotonicity are hard to detect---is consistent with our power and updating results in the following sections, however.
\end{rem}

\section{Frequentist testing}

We start by considering the possibility of frequentist testing against monotonicity. First,
we show that there exist frequentist tests with power against \emph{some} alternatives.

\begin{restatable}{prop}{prophaspower}
\label{prop:haspower}
\emph{(Frequentist tests exist)}
Suppose $n_U,n_T \geq 2$. Then there exist tests for monotonicity that control size and have power against some alternatives. That is, for any $\alpha \in (0,1)$ there exists a test $\delta: \supp(Y) \to [0,1]$ such that $\sup_{\theta_0 \in \Theta_0} E_{\theta_0}[\delta(Y)] \leq \alpha$ and $E_{\theta_1}[\delta(Y)] > \alpha$ for some $\theta_1 \in \Theta_1$.\footnote{The proof shows that there is a non-randomized test satisfying the conditions of the proposition for $\alpha$ sufficiently large.}      
\end{restatable}

\paragraph{\textbf{Intuition for tests.}} Observe that in \Cref{ex:2 units}, the support of $Y$ was $\{(0,1), (1,0)\}$ when there was 1 always-taker and 1 never-taker, but was $\{(0,0),(1,1)\}$ when there was 1 complier and 1 defier. This illustrates that the support of $Y$ may be different when monotonicity holds versus when it is violated. This general idea is the basis of the construction of tests for monotonicity in the proof to \Cref{prop:haspower}. We show that when $n \geq 4$, there always exists an alternative $\theta_1$ such that the support of $Y$ under $\theta_1$ does not contain the support of $Y$ under $\theta_0$ for any $\theta_0$ satisfying the null: i.e. $S^Y_{\theta_1} \not\supseteq S^Y_{\theta_0}$ for any $\theta_0 \in \Theta_0$, where $S^Y_\theta$ is the support of $Y$ under $\theta$.\footnote{Note that this is not the case in \Cref{ex:2 units}, because the support of $Y$ when there are two always-takers is $(1,1) \subseteq S^Y_{\theta_1} = \{(0,0),(1,1)\}$. There are thus no non-trivial tests of monotonicity with $n_T = n_U = 1$.} It follows that a test that rejects if and only if $Y \in S^Y_{\theta_1}$ has power of 1 to reject $\theta_1$, but only has size $\alpha_0 = \sup_{\theta_0 \in \Theta_0} E_{\theta_0}[ 1\{ Y \in S^Y_{\theta_1} \} ] < 1$. We can then construct a non-trivial randomized test with arbitrary size $\alpha \leq \alpha_0$ by rejecting with probability $\alpha/\alpha_0 > \alpha$ when $Y \in S^Y_{\theta_1}$.

Although non-trivial tests against monotonicity exist from the design-based perspective,
we expect their power to be poor in practice. Our next result shows that any test of
monotonicity has near-trivial power against some alternatives. Specifically, we show that if $n$ is even, there always exists an alternative for which power is no better than $\alpha$. For $n$ odd, we show there exists an alternative such that power is bounded above by $\alpha (1+ O(2^{-n}))$.\footnote{Surprisingly, there is a test for $n=13, n_T=6$ whose minimum power over $\Theta_1$ is ever-so-slightly larger than $\alpha$. For $\alpha=0.05$, such a test achieves power at least $0.05 + 2 \times 10^{-7}$ over all of $\Theta_1$.} This implies that there are no consistent tests of monotonicity (that is, tests for which power converges to 1 for all alternatives along a sequence of finite-populations with $n \to \infty$). 

\begin{restatable}{prop}{prophasnopower}
\emph{(Near-trivial power for some alternative)}
Suppose that $\delta: \supp(Y) \to [0,1]$ is a level-$\alpha$ test, i.e. $\sup_{\theta_0
\in \Theta_0} E_{\theta_0}[\delta(Y)] \leq \alpha$ for some $\alpha \in (0,1)$. If $n$ is
even, then $\delta$ has trivial power against some alternative: there exists $\theta_1 \in
\Theta_1$ such that $E_{\theta_1}[\delta(Y)] \leq \alpha$. If $n$ is odd, then there
exists $\theta_1 \in \Theta_1$ such that $E_{\theta_1}[\delta(Y)] \leq \alpha \left(1 + \tfrac{1}
{2^{n-1}-1} \right)$.
\end{restatable}

We can further show that the lack of power is \emph{generic} in the following sense: if we
construct a test to have power against some alternative $\theta_1$, then there are
alternatives
``near'' $\theta_1$ such that power is low. More precisely, around any $\theta_1$,
we can define a weighting function that mimics sampling from a population with type
frequencies $p_1 = \theta_1/n$. When $n$ is reasonably large, these weights are
``concentrated'' around $\theta_1$ in terms of type frequencies. Weighted average power (WAP) under this weighting
function turns out to never be larger than 2.51$\alpha$, uniformly over alternatives
$\theta_1$.
\begin{restatable}{prop}{propnowap}
\label{prop:nowap}
\emph{(Low WAP)}    Assume $n \geq 4$. Fix any $\theta_1 \in \Theta_1$ and let $v = \tfrac{1}{n} \min(\theta_{1,d},
\theta_{1,c}) \ge 1/n$ be
the
size of the monotonicity violation. Let $\vartheta \sim
\mathrm{Multinomial} (n, \theta_1/n)$. 
Consider the weight function $w_n(\cdot; \theta_1)$ over
$\Theta_1$ defined by the probability mass function of $\vartheta \mid (\vartheta \in
\Theta_1)$. Fix a test $\delta(\cdot)$ such that $\sup_{\theta_0 \in\Theta_0} E_{\theta_0}[\delta
(Y)] \le \alpha$. Then, the weighted average power around $\theta_1$ is bounded,
\[
    \mathrm{WAP}(\delta; \theta_1) := \sum_{\tilde{\theta}_1 \in \Theta_1} E_{\tilde\theta_1}[\delta
    (Y)] w_n
    (\tilde{\theta}_1;\theta_1) \le \alpha  \frac{1}{1-2(1-v)^n+(1-2v)^n} \le 2.51 \alpha.
    \numberthis
    \label{eq:wap}
\]
\end{restatable}

\Cref{prop:nowap} implies, for example, that a 5\% test never has weighted average power
larger than 12.6\%. The bound given in \Cref{prop:nowap} becomes even tighter if one
imposes a lower bound on the fraction of the finite population that are compliers/defiers.
For example, in a population of 100, if $\min(\theta_d, \theta_c) \ge 5$, then the upper
bound becomes $5.06\%$, which is virtually the same as size, uniformly over all such
alternatives.

Interestingly, while we typically expect statistical power to increase with $n$,
\Cref{prop:nowap} implies a \emph{tighter} upper bound on WAP the \emph{larger} is $n$
(holding fixed the share of compliers and defiers). Intuitively, having a larger finite
population is more similar to having an infinite superpopulation, in which case there is
no testable content of monotonicity.%
\footnote{
Similarly, the bound on WAP gets \emph{tighter} when the violation $v$ of the null hypothesis is \emph{larger}.}

\paragraph{\textbf{Proof sketch.}} Let $p_1 = \theta_1/n$ be the type shares under alternative $\theta_1$. Consider a Bayesian who believes the $n$ units to have been sampled from a superpopulation with proportions $p_1$, i.e. with prior over $\theta$ of $\tilde\pi_B \sim \mathrm{Multinomial}(n,p_1)$. The marginal distribution over $Y$ implied by this prior corresponds exactly to superpopulation sampling, $\tilde\pi_B(Y{=}y) = g_{p_1}(y)$. Now, consider a second Bayesian who believes the data to have been sampled from a superpopulation with shares $p_0$ satisfying monotonicity, i.e. with prior $\tilde\pi_A \sim \mathrm{Multinomial}(n,p_0)$. By analogous logic, the implied marginal distribution over $Y$ for this prior is $g_{p_0}(y)$. By \Cref{prop:id}, we can choose $p_0 \in \mathcal{P}_0$ such that $g_{p_0}(y) = g_{p_1}(y)$, and thus the two priors imply the same distribution over $Y$. It follows that the weighted average power under the two priors is the same,
\[E_{\theta \sim \tilde\pi_B}[ E_\theta[\delta(Y)] ] = E_{\theta \sim \tilde\pi_A}[ \underbrace{E_\theta[\delta(Y)] }_{\mathclap{\text{$\leq \alpha$ since $\theta \in \Theta_0$}}}] \leq \alpha, \]
where the inequality uses the fact that $\tilde\pi_A(\theta \in \Theta_0) = 1$ and $\delta$ controls size. Note, however, that under prior $\tilde\pi_B$, the design-based version of monotonicity holds with low probability: it is satisfied only if the sample of size $n$ from a superpopulation with both compliers and defiers happens to have zero compliers or zero defiers (or both), which occurs with probability \[(1-p_{1,d})^n + (1-p_{1,c})^n - (1-p_{1,d}-p_{1,c})^n \approx 0.\]  It follows that expected power under $\tilde\pi_B$ conditional on monotonicity being violated ($\theta \in \Theta_1$) is close to the unconditional expectation,
\[E_{\theta \sim \tilde\pi_B}[ E_\theta[\delta(Y)] \mid \theta \in \Theta_1 ] \approx E_{\theta \sim \tilde\pi_B}[ E_\theta[\delta(Y)] ] , \] and hence
$$ \mathrm{WAP}(\delta; \theta_1) = E_{\theta \sim \pi_B}[ E[\delta(Y)] \mid \theta \in \Theta_1 ] \approx  E_{\theta \sim \pi_B}[ E[\delta(Y)]  ]  \leq \alpha. $$
The proof formalizes this argument by providing an upper bound on the approximation error in the argument above, using the fact that $p_{1,d},p_{1,c} \geq 1/n$.

\paragraph{\textbf{Numerical illustration with $n=30$.}} To illustrate these results, \cref{fig:powerhist} computes the most powerful 5\%-level
test for every possible alternative $\theta \in \Theta_1$, obtained via linear
programming, for $n = 30$ and $n_T = 15$. Perhaps surprisingly, all 4495 alternatives are
testable, in the sense that for each alternative, there exists a test targeted to that
alternative with power more than 0.05. However, the power for these tests for the targeted
alternative tends to be modest: among these tests engineered to maximize power at a given
alternative, only 53 of 4495 alternatives have tests with power above 0.15, and all of
them have power below 0.31. These tests also tend to have very poor power for nearby
alternatives, as suggested by \Cref{prop:nowap}: Only 18 have weighted average power above the nominal threshold $0.05$, and the maximum WAP is a measly $0.0567$. As a specific example, the optimal test against 18 defiers and 12 compliers achieves the maximal power of $0.31$, but this test rejects with probability zero when there are 17 defiers and 13 compliers.

\begin{figure}[!ht]
    \includegraphics[width=\textwidth]{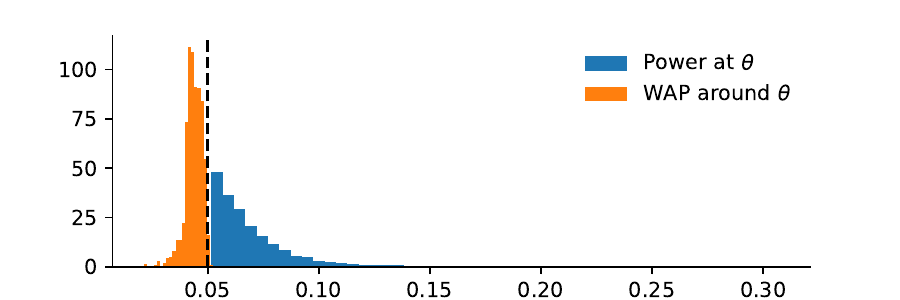}
    \caption{Power of the most powerful test for a given alternative $\theta \in
    \Theta_1$ for $n=30$, $n_T=15$}

    \begin{proof}[Notes]
        We construct the most powerful 5\%-test for a given alternative $\theta$, which we
        obtain via linear programming: $\max_{\delta: \supp(Y) \to [0,1]} \sum_y f_{\theta}
        (y) \delta (y)$ subjected
        to $\sum_y f_\vartheta(y) \delta(y) \le 0.05$ for all
        $\vartheta\in\Theta_0$. The
        blue patches show the distribution of power at $\theta$ for each test, across all
        4495 values in $\Theta_1$. For each test, we also compute its weighted average
        power \eqref{eq:wap} and show its distribution in orange. 
    \end{proof}
    \label{fig:powerhist}
\end{figure}

\paragraph{\textbf{Unbiased tests.}} We can obtain even sharper limits on power if we restrict attention to unbiased tests. Recall that a test is unbiased if its power against all alternatives is weakly greater than its size. \Cref{prop:unbiasedexists} below shows that unbiased tests for monotonicity do exist (at least when $n_T=n_U$). However, \Cref{prop:unbiased} implies that unbiased tests have asymptotically trivial power.

\begin{restatable}{prop}{propunbiasedexists} \emph{(Unbiased tests exist)}
\label{prop:unbiasedexists}
Suppose $n_T=n_U \geq 2$. Then there exists a non-trivial unbiased test of monotonicity: for any $\alpha \in (0,1)$, there exists $\delta: \supp(Y) \to [0,1]$ such that $\sup_{\theta_0 \in \Theta_0} E_{\theta_0}[ \delta(Y) ] \leq \alpha \leq \inf_{\theta_1 \in \Theta_1} E_{\theta_1}[ \delta(Y) ]$, with $E_{\theta_1}[ \delta(Y) ] > \alpha$ for at least one $\theta_1 \in \Theta_1$. 
\end{restatable}

\begin{restatable}{prop}{propunbiased} \emph{(Unbiased tests have asymptotically trivial power)}
\label{prop:unbiased}
    Fix $\epsilon > 0$ and $n \geq 4$. Fix any  $\theta \in \Theta_{1}$ such that $v = \min
    (\theta_{d}, \theta_{c})/n \ge \epsilon$. Let $\delta$ be any unbiased level-$\alpha$ test, i.e. a test satisfying $\sup_{\theta_0 \in \Theta_0} E_{\theta_0}[\delta(Y)] \le \alpha \le \inf_{\theta_1 \in
        \Theta_1} E_{\theta_1}[\delta(Y)].$
    Then we have that
    \[
        E_{\theta}[\delta(Y)] \le \alpha (1 + \eta_n(\epsilon))
    \]
    for $\eta_n(\epsilon) = 3.125 n^{1.5} (1-\epsilon)^n$.
\end{restatable}

In particular, \Cref{prop:unbiased} implies that along any sequence of finite populations with $v_n = \min(\theta_{d,n}, \theta_{c,n})/n \ge \epsilon$, the power of any unbiased test is $(1+o(1)) \alpha$ as $n \rightarrow \infty$.

The proof of \Cref{prop:unbiased} casts the problem of maximizing $E_{\theta}[\delta(Y)]$
subject to unbiasedness and size control as a linear program. The dual of this
program---whereby any feasible value implies an upper bound for the power of
$\delta$---involves choosing certain weights over the null and the alternative. The weights
$w_n$ defined in \Cref{prop:nowap} turn out to enable a nontrivial upper bound of the primal value.

\section{Bayesian updating}

We next ask whether Bayesians update about whether $\theta\in\Theta_0$. The following result shows that \emph{some} Bayesians update, but there exist Bayesians who find the data totally uninformative about whether monotonicity holds, in the sense that their posterior on $\theta \in \Theta_0$ is equal to their prior almost surely.

\begin{restatable}{prop}{bayescombined}
\label{prop:bayescombined}
\emph{(Bayesian updating)}
\begin{enumerate}
    \item 
    Some Bayesians update: there exists a prior $\pi$ with $\pi
    (\theta \in \Theta_0) \in (0,1)$ such that $\pi(\theta \in \Theta_0 \mid Y) \neq \pi(\theta \in \Theta_0)$ with
    positive $\pi$-probability.

    \item
    There exist (nontrivial) Bayesian priors over $\theta$ that never update about the probability that monotonicity holds: if $n$ is even, then for any $c \in (0,1)$, there exists a prior distribution $\pi$ over $\theta$ such that $\pi(\theta \in \Theta_0 \mid Y) = \pi(\theta \in \Theta_0) =c$, $\pi$-almost surely.
\end{enumerate}
\end{restatable}

We note that the conclusion in part 1 that \emph{some} Bayesian updates is rather weak. In fact, even in the superpopulation setting, there is \emph{some} Bayesian who updates: consider, for example, the Bayesian who believes that there are defiers if and only if the average treatment effect is larger than 0.5; this Bayesian updates about the validity of monotonicity based on the information in the data about the average treatment effect. However, the updating in the design-based setting is slightly less trivial: a Bayesian in fact updates about the relative probability of two type vectors $\theta_0, \theta_1$ that imply the same marginal distributions of $y(0)$ and $y(1)$ in the finite population, whereas this is not true in the superpopulation setting. 

Part 2 of \Cref{prop:bayescombined} implies that some Bayesians do not find the data at all informative about whether monotonicity holds. Hence, the data will not be persuasive to every audience member. This also suggests that Bayesian inference in this setting, as considered in \citet{ding_model-free_2019} and \citet{christy_counting_2025}, for example, will necessarily be sensitive to the choice of prior: for the class of priors in the second part of \Cref{prop:bayescombined}, posterior statements about monotonicity simply match the priors. 

The existence of a prior that does not update also implies that any binary classifier of whether $\theta \in \Theta_0$ does no better than random guessing for some parameter value $\theta$. Specifically, given $a \in \{0,1\}$, let $\ell(a,\theta) = 1\{a \neq 1\{\theta \in \Theta_0 \} \}$ be an indicator for whether $a$ mis-classifies whether $\theta \in \Theta_0$ (i.e. zero-one classification loss). Let $c: \mathcal{Y} \to [0,1]$ be a possibly-randomized classifier that sets $a=1$ with probability $c(Y)$, and let $R(c,\theta) = E_\theta[ c(Y)\ell(1,\theta) + (1-c(Y) \ell(0,\theta)]$ be its average classification error at parameter $\theta$ (risk under 0--1 loss).

\begin{restatable}{cor}{trivialclassification}
 \label{cor:trivialclassification}
Suppose $n$ is even. Then $\inf_{c(\cdot)} \sup_{\theta} R(c,\theta) = 0.5$.         
\end{restatable}

\noindent Note that a trivial classifier that guesses randomly ($c(y) = 0.5,$ $\forall y$) achieves a 0.5 mis-classification rate. \Cref{cor:trivialclassification} implies that \emph{every} classifier does no better than this at some parameter value $\theta$ (when $n$ is even). This follows from Part 2 of \Cref{prop:bayescombined}: there exists a prior that never updates about whether $\theta \in \Theta_0$, and hence Bayes risk under this prior must be trivial. \citet{christy_counting_2025} consider the maximum-likelihood estimator $\hat\theta^{\text{MLE}}$ and say that it ``provides evidence in favor of monotonicity'' if $\min(\hat\theta^{\text{MLE}}_d, \hat\theta^{\text{MLE}}_c) = 0$. \Cref{cor:trivialclassification} implies that this classification of whether the data support monotonicity does no better than random guessing at some parameter value.%
\footnote{In fact, in some cases, this classification rule does worse than random guessing. Consider \Cref{ex:2 units} above with $n=2$. Note that with 1 complier and 1 defier, the support of $Y$ is $\{(0,0),(1,1)\}$, with each support point occurring with probability $1/2$. However, $Y=(1,1)$ with probability 1 if there are 2 always-takers, and $Y=(0,0)$ with probability 1 if there are two never-takers. Thus, the MLE never corresponds to having 1 complier and 1 defier, so worst-case misclassification error is 1. Moreover, we show in \Cref{prop:mleconvergencesimpler} below that the MLE converges in probability to a limit where at least one type has zero share, even if all types of have positive share in the population.}
Moreover, we cannot ``fix'' this undesirable property by choosing a different classification rule.

Part~2 of \Cref{prop:bayescombined} and \Cref{cor:trivialclassification} focused on the
case where $n$ is even. In the appendix, we show that when $n$ is odd, there exist
non-degenerate priors that \emph{minimally} update in the sense that the expected change
between the prior and posterior is small: \Cref{prop:bayesupdateodd} shows that there
exists a prior such that \[ E_{Y \sim \pi_Y}[ | \pi(\theta \in \Theta_0 \mid Y) -
\pi(\theta \in \Theta_0) |] = O(2^{-n}), \] where $\pi_Y$ is the distribution of $Y$
induced by prior $\pi$. This implies \Cref{cor:classification-odd}, which states that when
$n$ is odd, a classifier of whether $\theta \in \Theta_0$ has worst-case
mis-classification error at least $0.5 - O(2^{-n})$.

\section{Extension: Estimating the extent to which monotonicity is violated}

We have shown that despite $\theta$ being formally identified, there is limited ability in practice to test whether monotonicity is violated, i.e. to test whether $v = \frac{1}{n} \min(\theta_d,\theta_c) > 0$. We now show that it is also difficult to learn about the \emph{extent} of the monotonicity violation---i.e. it is difficult to estimate $v$ (not just determine whether it is zero), and likewise, it is difficult to estimate the full type vector $\theta$. Our first result provides lower bounds on the achievable mean-squared-error (MSE) for estimating the monotonicity violation $v$ and the full vector $\theta$.

\begin{restatable}{prop}{estimation}
\label{prop:estimation}
\emph{(Estimation)}
\begin{enumerate}
    \item \emph{(Limited learning about the monotonicity violation)}
        Let $v = \tfrac{1}{n} \min(\theta_d,\theta_c)$ be the size of the monotonicity violation, and $\hat{v}$ an estimator of $v$. Then
        \begin{align*}
            \max_{\theta \in \Theta} E_\theta [ (\hat{v} - v )^2] \geq \frac{1}{16}  \left(1 - \frac{1}{\sqrt{n}}\right)^2.
        \end{align*}

    \item \emph{(Limited learning about full type proportions)}
        Let $\hat{\theta}$ be an estimator of $\theta$.
        Then
        \begin{align*}
            \max_{\theta \in \Theta} E_\theta\bigg[ \big\| \tfrac{1}{n} (\hat{\theta} - \theta) \big\|^2 \bigg] \geq  \frac{1}{4} + \frac{1}{4 n}.
        \end{align*}
    \end{enumerate}
\end{restatable}

The first part of the proposition shows that we cannot estimate the monotonicity violation $v$ much better than by the trivial guess of $\hat{v} = \tfrac{1}{4}$, which achieves a worst-case risk of $\tfrac{1}{16}$. Indeed, the proof of this result establishes a slightly tighter bound that shows, for example, that the minimax risk for estimating $v$ is at least 74\% that of the trivial estimator for $n \geq 30$ and at least 85\% for $n \geq 100$. The second part of the proposition shows that the minimax MSE for estimating the type proportions $\frac{\theta}{n}$ is always at least $1/4$. We note that the mid-point of the (sample) Fr\'echet--Hoeffding bounds has worst-case MSE converging to $1/4$ as $n \to \infty$, and thus is asymptotically minimax; this suggests that in reasonably large samples we cannot do much better than guessing the Fr\'echet--Hoeffding midpoint.

Interestingly, while the estimator that takes the midpoint of the Fr\'echet--Hoeffding bounds is asymptotically minimax, the maximum likelihood estimator (MLE), as suggested by \citet{christy_counting_2025}, can actually be shown to converge in probability to one of the Fr\'echet--Hoeffding endpoints. This implies that the large-$n$ worst-case MSE of the MLE is worse than the minimax bound given above. Specifically, let
\[
\hat{\theta}^{\text{MLE}} \in \arg\max_{\theta \in \Theta} f_\theta(Y_T, Y_U)
\]
be a maximum-likelihood estimator (the way ties are broken is immaterial for our result). The following result shows that the MLE estimate of the fraction of types converges to a specific endpoint of the Fr\'echet--Hoeffding bounds.

\begin{restatable}[Probability limit of MLE]{prop}{mleconvergencesimpler}
\label{prop:mleconvergencesimpler}
Consider a sequence of type counts $\theta_n$ with $\tfrac{1}{n} \theta_n \rightarrow p
\in [0,1]^4$ along the sequence, and let $q_U = p_{at} + p_{d}$, $q_T = p_{at} + p_c$.
Assume that $n_T/n \in [\lambda, 1-\lambda]$ for some fixed $\lambda \in (0, 1/2)$. 
Suppose without loss of generality that $q_T \geq q_U$ (so the treatment
effect is positive), and define the limiting Fr\'echet--Hoeffding bounds by $$p^L := (q_U, 1{-}q_T, 0, q_T {-} q_U),  \quad p^U := (q_U{-}v_{max}, 1{-}q_T{-}v_{max}, v_{max}, q_T {-} q_U {+} v_{max}),$$ for $v_{max} = \min(q_U,1{-}q_T)$. If $0 {<} q_U,q_T {<} 1$, $q_U, q_T \neq \tfrac{1}{2}$, $q_U \neq q_T$, and $q_U {+} q_T \neq 1$, then
\[
\frac{1}{n}\hat{\theta}^{\mathrm{MLE}}
\pto  p^* = p^*(q) =
\begin{cases}
p^L,
& \text{if } q_T, q_U > \tfrac{1}{2} \text{ or } q_T, q_U < \tfrac{1}{2}, \\[6pt]
p^U,
& \text{if } q_U < \tfrac{1}{2} < q_T.
\end{cases}
\]
Under either limit, $\min(p^*_{at},p^*_{nt},p^*_d,p^*_c) = 0$.
\end{restatable}

\Cref{prop:mleconvergencesimpler} establishes that as $n \to \infty$, the MLE estimator converges in probability to one of the Fr\'echet--Hoeffding bounds. Which of the bounds it converges to depends only on $q_T,q_U$, the asymptotic shares of individuals with $Y(1)=1$ and $Y(0)=1$, respectively. If these shares are on opposite sides of $\frac{1}{2}$, then the MLE converges to the Fr\'echet--Hoeffding upper bound (under which either always-takers or never-takers have zero share) while if they are on the same side, it converges to the Fr\'echet--Hoeffding lower bound (under which either compliers or defiers have zero share).%
\footnote{
    This pattern is in line with the empirical from Figure B.1 in \citet{christy_counting_2025}, which computes and visualizes the maximum-likelihood estimator for $n \in \{50,200\}$ with $n_T = n_U = \tfrac{n}{2}$. The conditions on $q$ ensure we are not on the transition boundary where the MLE may not have a unique limit.}
Thus, whether or not the MLE suggests that there is a violation of monotonicity asymptotically depends only on the \emph{marginal} distributions of the potential outcomes, as formalized in the following corollary.

\begin{restatable}{cor}{mleviolation}
\label{cor:mleviolation}
\emph{(MLE for monotonicity violation)}
Denote by $\hat{v}^{\text{MLE}} = \tfrac{1}{n}\min(\hat{\theta}^{\text{MLE}}_d, \hat{\theta}^{\text{MLE}}_c)$ the estimate of the monotonicity violation implied by $\hat{\theta}^{\text{MLE}}$.
Then, under the conditions of \Cref{prop:mleconvergencesimpler},
\begin{align*}
    \hat{v}^{\text{MLE}} \xrightarrow{p} v^*(q_T,q_U)
    =
    \begin{cases}
        0, & \text{if } q_T, q_U > \tfrac{1}{2} \text{ or } q_T, q_U < \tfrac{1}{2}, \\
        v_{max} > 0, & \text{if } q_U < \tfrac{1}{2} < q_T.
    \end{cases}
\end{align*}

\end{restatable}

For example, 
if the average treatment effect of $D$ on $Y$ is larger than $\tfrac{1}{2}$, then $q_U < \tfrac{1}{2} < q_T$, and the maximum-likelihood estimator assigns positive shares to both compliers and defiers in the limit.
This implies, counterintuitively, that the MLE will asymptotically suggest that a treatment negatively impacts some people whenever its average effect is sufficiently positive.
Moreover, \Cref{prop:mleconvergencesimpler} implies that the MLE asymptotically assigns zero share to one of the types even if all four types are present in the population. With large $n$, the presence or absence of a particular type in the MLE is thus purely a function of the marginals and does not imply that any type is actually present or absent.

The above results also imply that the worst-case risk of the estimators $\hat{\theta}^{\text{MLE}}$ of type counts and $\hat{v}^{\text{MLE}}$ of the monotonicity violations are asymptotically worse than the minimax bounds.
Specifically, the worst-case MSEs have limits
\begin{align*}
    \liminf_{n\to\infty}\max_{\theta \in \Theta} E_\theta\left[(\hat{v}^{\text{MLE}} -
    v)^2\right] &\ge \tfrac{1}{4},
    &
    \liminf_{n\to\infty}\max_{\theta \in \Theta} E_\theta\left[\big\|\tfrac{1}{n}(\hat{\theta}^{\text{MLE}} -
    \theta)\big\|^2\right] &\ge 1
\end{align*}
that are at least four times those of
always guessing $\hat{v} = \tfrac{1}{4}$ for $v$ ($\tfrac{1}{4}$ versus $\tfrac{1}{16}$)
and of
guessing the midpoint of the (sample) Fr\'echet--Hoeffding bounds for $\theta$ ($1$ vs $\tfrac{1}{4}$).

The result in \Cref{prop:mleconvergencesimpler} was anticipated by \citet{copas_randomization_1973}, who observed that in large samples, the likelihood is approximately Gaussian with a mean vector that depends on $\theta$ only through the sums $\theta_{at} + \theta_d$ and $\theta_{at} +\theta_c$. Treating the normal approximation as exact, \citet{copas_randomization_1973} then argued that the MLE would land at one of the Fr\'echet--Hoeffding endpoints. \Cref{prop:mleconvergencesimpler} turns this heuristic into an exact statement, establishing convergence in probability and identifying which endpoint is selected. Specifically, the proof formalizes the extent to which the probability mass function (PMF) $f_\theta(y_T, y_U)$ for the vector $(Y_T, Y_U)$ can be approximated by the PMF of the discrete normal approximation
\begin{align}
\label{eqn:Normalapprox}
    \colvec{2}{Y_T}{Y_U}
    &\stackrel{\approx}{\sim}
    \mathcal{DN}
    \left(
        \mu_\theta,
        \Sigma_\theta
    \right),
\end{align}
where 
\begin{align*}
    \mu_\theta &=
    \tfrac{1}{n}
    \begin{psmallmatrix}
            n_T (\theta_{at} + \theta_c) \\
            n_U (\theta_{at} + \theta_d)
    \end{psmallmatrix},
    &
    \Sigma_\theta &=
    \tfrac{n_T n_U}{n^2 (n-1)}
    \begin{psmallmatrix}
            (\theta_{at} + \theta_c)(n - \theta_{at} - \theta_c) & (\theta_{at} + \theta_c)(\theta_{at} + \theta_d) - n \theta_{at} \\
            (\theta_{at} + \theta_c)(\theta_{at} + \theta_d) - n \theta_{at} & (\theta_{at} + \theta_d)(n - \theta_{at} - \theta_d)
    \end{psmallmatrix}
\end{align*}
are the mean and variance of $\begin{psmallmatrix} Y_T \\ Y_U \end{psmallmatrix}$ under randomization, and $\mathcal{DN}
    (\mu, \Sigma)$ is the distribution on $\mathbb{Z}^2$ whose PMF at $x \in \mathbb{Z}^2$ is proportional to the PDF of $\Norm(\mu, \Sigma)$ at $x$. 

This normal approximation clarifies the dependence of the likelihood on the full vector
$\theta$ of types. The mean and variance of each of the $Y_T,Y_U$ only depend on the
\emph{marginal} counts $\theta_{at} + \theta_c$, $\theta_{at} + \theta_d$. Only the
covariance provides separate information about the individual type shares. To the extent
that the normal approximation is accurate, learning about the individual type shares is
equivalent to the problem of learning the covariance of a bivariate (discretized) normal
distribution \emph{from a single draw}.

In order to derive the asymptotic distribution of the maximum-likelihood estimator, it is
not sufficient to argue that the approximation in \eqref{eqn:Normalapprox} merely holds in
a convergence-in-distribution sense.
Indeed, our proof involves showing that the approximation holds \emph{uniformly} over an
appropriate set of type counts $\theta$. In doing so, we show that our experiment based on
total randomization and the discretized normal experiment based on 
\eqref{eqn:Normalapprox} are asymptotically equivalent in a limit-of-experiment sense,
which we make precise in \Cref{prop:lecam} in the supplemental appendix.

\section{Conclusion}

We study what a completely randomized experiment reveals about finite-population monotonicity. We show that from the design-based perspective, the type counts $\theta$ in the finite population are in fact identified. However, the extent to which we can feasibly learn about violations of monotonicity is severely limited: frequentist tests generically have poor power, and some Bayesians never update about whether the null is true. Likewise, estimates of the magnitude of the violation of monotonicity or the full type vector are generically poor. Thus, formal identification translates to little, if any, practical learning about monotonicity. These results highlight that conclusions about identification may differ depending on whether one adopts a sampling-based versus design-based perspective, and that studying the properties of frequentist tests and Bayesian updating may provide a more realistic assessment of the extent to which learning is possible in design-based settings. An interesting avenue for future research is to explore whether similar issues arise in other design-based causal-inference problems. 

\newpage

\bibliographystyle{aer}
\bibliography{Bibliography}

\newpage

\appendix
\renewcommand{\sectionname}{}
\section{Appendix (proofs)}

\propid*
\begin{proof}
For the first claim, it suffices to construct a mapping from $f_\theta$ to $\theta$. 
Fix $q = n_T/n$ as the treatment probability. 
    Observe that, by linearity of expectation, \[
        E_{\theta}[(Y_T, Y_U)'] = (n_{at} q + n_c q, (1-q) n_{at} + (1-q) n_d)'. 
    \]
    Second, label underlying units $i=1,\ldots,n$ and collect always takers in $N_{at}
    \subset [n]$, compliers in $N_c \subset [n]$, and defiers in $N_d \subset [n]$. Let
    $N_T \subset [n]$ collect the (random) set of treated individuals. 
    Observe that \begin{align*}
    &E_{\theta}[Y_TY_U] = \sum_{i \neq j, i\in N_{at} \cup N_c, j \in N_{at} \cup N_d} P
    (\text{$i \in N_T$, $j \not \in N_T$}) \tag{Linearity of expectations} \\
    &= |\{(i,j) \,:\, i \neq j, i\in N_{at} \cup N_c, j \in N_{at} \cup N_d \} | \cdot P
    (\text{$1 \in N_T$, $2 \not \in N_T$}) \tag{Symmetry}\\
    &= \bk{(n_{at} + n_c) (n_{at} + n_d) - n_{at}} P
    (\text{$1 \in N_T$, $2 \not \in N_T$}).
    \end{align*}
    Thus,
    for $r = P
    (\text{$1 \in N_T$, $2 \not \in N_T$}) = \tfrac{n_T n_U}{n (n-1)}$,
    \[
        (n_{at} + n_c)(n_{at} + n_d) - n_{at} = \frac{E_{\theta}[Y_T Y_U]}{r} 
    \]
    is a functional of $f_\theta$. Letting $(\mu_T,\mu_U, \mu_{TU}) := (E_\theta[ Y_T] , E_\theta[Y_U], E_\theta[Y_T Y_U])$, we see that 
\begin{equation*}
    \frac{\mu_T}{q} = n_{at} + n_c, \hspace{1cm} \frac{\mu_U}{1-q} = n_{at} + n_d, \hspace{1cm} \frac{\mu_{TU}}{r} = \frac{\mu_T}{q}\,\frac{\mu_U}{1-q} - n_{at}.
\end{equation*}
It follows that
\begin{equation*}
    n_{at} = \frac{\mu_T}{q}\,\frac{\mu_U}{1-q} - \frac{\mu_{TU}}{r},  \hspace{1cm } n_c = \frac{\mu_T}{q} - n_{at}, \hspace{1cm} n_d = \frac{\mu_U}{1-q} - n_{at}.
\end{equation*}

\noindent Since $(\mu_T,\mu_U, \mu_{TU})$ is a function of $f_\theta$, we have therefore shown that there is an inverse mapping from $f_\theta$ to $\theta$, and hence the mapping $\theta \mapsto f_{\theta}$ must be one-to-one, as needed. 

For the second claim, observe that the superpopulation data-generating process can be
    represented as: 
    \begin{enumerate}
        \item Fix index $i = 1,\ldots, n$. The first $n_T$ units are treated.
        \item For each unit $i$, sample her type from $\mathrm{Multinomial}(1, p).$
        \item The counts $Y$ are then functions of type counts in the first $n_T$ units
        and the next $n-n_T$ units. 
    \end{enumerate}
    We can check that under this process, $\theta \sim \mathrm{Multinomial}(n, p)$ and
    $Y \mid \theta \sim f_\theta$. But under this process, we have that $Y_T \sim \Bin
    (n_T, p_ {at} + p_c)$ and $Y_U \sim \Bin(n-n_T, p_{at}
    +
    p_d)$, and $Y_T \indep Y_U$. We thus see that the unconditional distribution of $Y$ depends on $p$ only through the sums $p_{at} + p_c$ and $p_{at} + p_d$. (\citet[][p. 472-3]{copas_randomization_1973} likewise shows that the likelihood depends only on these sums.) For any given $p \in \mathcal P_1$, observe that $p'$ defined by \[
        p'_{at} = p_{at} + \min(p_d, p_c) \quad p_{nt}' = p_{nt} + \min(p_d, p_c) \quad p_d' =
        p_d - \min(p_d, p_c) \quad p_{c}' = p_c - \min(p_d, p_c)
    \]
    generates observationally equivalent $(Y_T, Y_U)$. However, $\min(p'_d, p'_c) = 0$ and
    thus $p' \in \mathcal P_0$. 
\end{proof}

\prophaspower*

\begin{proof}
Suppose, without loss of generality, that $n_T \leq n_U$ (if not, we can adopt the same argument reversing the roles of $n_d$ and $n_c$). Let $\theta_1$ be the type with $n_d = n_T - 1$, $n_c = n-n_d $, and $n_{at}=n_{nt} = 0$. Note that $n_c = n-n_T +1 \geq 2$. Observe that the support of $Y$ under $\theta_1$ is $$S_{\theta_1}^Y = \{ (m,n_d-(n_T-m)) \,:\, m=1,..., \min(n_T, n_c) \},$$ with $(m,n_d-(n_T-m))$ corresponding to the realization of $Y$ when $m$ compliers are assigned to treatment. Observe that the points in $S_{\theta_1}^Y$ lie on an upward sloping line with slope of 1. That is, for any $y= (y_1,y_0)$ and $y' =(y_1',y_0')$ both in $S^Y_{\theta_1}$, we have that $y-y' = c \cdot (1,1)$. 

Now, for any $\theta_0 \in \Theta_0$, let $S^Y_{\theta_0}$ denote the support of $Y$ under $\theta_0$. We claim that $S^Y_{\theta_0} \not\subseteq S^Y_{\theta_1}$ for all
$\theta_0 \in \Theta_0$. 

\begin{enumerate}
    \item [Case 1] Suppose that $\theta_0$ corresponds to there only being
one type in the population. Then $S_{\theta_0}^Y$ is a singleton set with element either $
(0,0)$, $(n_T,0)$, $(0,n_U)$ $(n_T,n_U)$ depending on whether the lone type is $nt$, $c$, $d$ or $at$. It is
clear that $(0,0) \not\in S_{\theta_1}^Y$ and $(0,n_U) \not\in S_{\theta_1}^Y$ since the first coordinate of points in $S_{\theta_1}^Y$ is strictly positive by construction. Next, note that the only possible element of
$S_{\theta_1}^Y$ with first element equal to $n_T$ is $(n_T,n_d)$, corresponding to the
case where $m = n_T = \min(n_T,n_c)$. However, by construction we have that $0 < n_d =
n_T - 1 < n_U$, and thus $(n_T,n_d) \neq (n_T,0)$ and $(n_T,n_d) \neq (n_T,n_U)$.

 \item [Case 2] Suppose $\theta_0 \in \Theta_0$ that has positive numbers of at least two
 types.
Suppose towards contradiction that $S^Y_{\theta_0} \subseteq S^Y_{\theta_1}$. Consider any support point $(y_1,y_0)$ in $S^Y_{\theta_0}$. Note that $(y_1,y_0) = (\sum_i D_i y_i(1), \sum_i (1-D_i) y_i(0))$ for some choice of $D_i$ and a population such that $\{(y_i(1),y_i(0)) \}_{i=1}^n$ has the types with frequencies given by $\theta_0$. Since there are at least two types under $\theta_0$, there are distinct indices $j,k$ such that $D_j =1, D_k =0$ and individuals $j$ and $k$ are of different types, i.e. $(y_j(0),y_j(1)) \neq (y_k(0),y_k(1))$. Now consider the treatment assignment $\tilde{D}$ that swaps the treatment assignments of units $j$ and $k$ and leaves the other assignments unchanged, i.e. $\tilde{D}_j =0, \tilde{D}_k = 1$, and $\tilde{D}_i = D_i$ for $i \not\in \{j,k\}$. The realized outcome under treatment assignment $\tilde{D}$ is $(y_1',y_0') = (\sum_i \tilde{D}_i y_i(1), \sum_i (1-\tilde{D}_i) y_i(0))$, and thus $(y_1',y_0') \in S^{Y}_{\theta_0} \subseteq S^{Y}_{\theta_1}$. However, we have that
$$(y_1',y_0') - (y_1,y_0) = (y_k(1) - y_j(1), y_j(0) - y_k(0) ) .$$
\noindent Note that since $(y_j(0),y_j(1)) \neq (y_k(0),y_k(1))$, it follows that $(y_1',y_0') - (y_1,y_0) \neq 0$. Then since $(y_1,y_0),(y_1',y_0') \in S^Y_{\theta_1}$, it must be the case that $(y_1',y_0') - (y_1,y_0) = c \cdot (1,1)$ for an integer $c \neq 0$. It follows that 
\begin{align*}
 & y_k(1) - y_j(1) = c,
 &
 & y_j(0) - y_k(0) = c.
\end{align*}
\noindent Adding the two equations, we obtain that 
$$(y_k(1) - y_k(0)) - (y_j(1)-y_j(0)) = 2c \neq 0 .$$
\noindent Since the treatment effects $(y_k(1) - y_k(0))$ and $(y_j(1)-y_j(0))$ are each in $\{-1,0,1\}$, their difference can equal an even non-zero integer only if one is equal to 1 and the other is equal to $-1$. However, this contradicts $\theta_0$ satisfying monotonicity. It follows that $S^Y_{\theta_0} \not\subseteq S^Y_{\theta_1}$, as we wished to show. 

\end{enumerate}

Now, let $\delta^*(Y) = 1\{ Y \in S^Y_{\theta_1} \}$. Since $S^Y_
{\theta_0} \not\subseteq S^Y_{\theta_1}$ for all $\theta_0 \in \Theta_0$, it follows that
$E_{\theta_0}[\delta^*(Y)] < 1$ for all $\theta_0 \in \Theta_0$. Since $\Theta_0$ is
finite, we thus obtain that $\sup_{\theta_0 \in \Theta_0} E_{\theta_0}[\delta^*(Y)] < 1$.
However, by construction $E_{\theta_1}[\delta^*(Y)] = 1$. Hence, $\delta^*(Y)$ controls
size at level $\alpha^* = \sup_{\theta_0 \in \Theta_0} E_{\theta_0}[\delta^*(Y)] < 1$ and
has power of 1 against the alternative $\theta_1$. If $\alpha^* \leq \alpha$, then the proof
is complete by setting $\delta(Y) = \delta^*(Y)$ (and thus we have a non-randomized
test). If $\alpha^* > \alpha$, set $\delta(Y) = \tfrac{\alpha}{\alpha^*} \delta^*
(Y)$. Then by construction $\sup_{\theta_0 \in \Theta_0} E_{\theta_0}[\delta
(Y)] = \alpha$ and $E_{\theta_1}[\delta(Y)] = \tfrac{\alpha}{\alpha^*} > \alpha$. 
\end{proof}

\begin{lem} \label{lem: there exist equiv priors on the null and alt}
Suppose $n$ is even. Then there exist priors $\pi_A$ and $\pi_B$ on $\theta$
such
that $\pi_A(\theta \in \Theta_0) = 1$, $\pi_B(\theta \in \Theta_0) = 0$, and $\pi_A(Y{=}y)
= \pi_B(Y{=}y)$ for all $y \in \{0,\ldots,n_T\} \times \{0,\ldots, n_U\}$.     
\end{lem}
\begin{proof}
Let $\tilde{\pi}_p$ denote the multinomial prior on $(n_{at},n_{nt},n_{d},n_{c})$ with
probabilities $p = (p_{at},p_{nt},p_{d},p_{c})$. That is, \[
(n_{at},n_{nt},n_{d},n_{c}) \sim \tilde\pi_p \sim \mathrm{Multinomial}(n, p).
\] 
Define $\tilde{\pi}_A = \tilde{\pi}_{(0.5,0.5,0,0)}$ and $\tilde{\pi}_B =
 \tilde{\pi}_{(0,0,0.5,0.5)}$. The prior $\tilde{\pi}_A$ arises from assuming that the
 units in the finite population were sampled i.i.d. from a superpopulation in which half
 of individuals are always-takers and the other half are never-takers. The prior
 $\tilde{\pi}_B$ arises analogously if the superpopulation is half compliers and half
 defiers. As argued in the proof to \Cref{prop:id}, the marginal distribution of $Y$ after (i) sampling $\theta \sim \mathrm{Multinomial}(1, p)$ and (ii) generating $Y \sim f_{\theta}$ depends on $p$ only through $p_{at}+p_c$ and $p_{at} + p_d$. It follows that the two priors imply the same unconditional distribution for $Y$, $\tilde\pi_A(Y{=}y) = \tilde\pi_B(Y{=}y)$ for all $y$.

Now, consider type counts $\theta_0(n_{at}) = (n_{at}, n-n_{at}, 0,0)$, which have positive numbers of only always-takers and never-takers. Observe that the support
of $Y$ under $\theta_0(n_{at})$ is
\begin{equation}
  S^Y_{\theta_0(n_{at})} = \{ (m,n_{at}-m) \,:\, m = \max(0, n_{at}-n_U) ,\ldots,\min(n_T,n_{at}) \}, \label{eqn:s0}  
\end{equation} where $(m,n_{at}-m)$ is the realization of $Y$ if
$m$ of the always-takers are selected for treatment. Hence the support points for $Y$ under
$\theta_0$ lie on a line with slope of $-1$ and intercept $n_{at}$. Letting $S^Y_{\theta_0}$
denote the support of $Y$ under $\theta_0$, we see that the $S^Y_{\theta_0}$ are disjoint
for all $\theta_0$ to which $\tilde{\pi}_A$ assigns positive support: \[
S^Y_{\theta_0(n_{at})} \cap S^Y_{\theta_0(n_{at}')} = \emptyset \text{ for $n_{at} \neq
n_{at}'$}
.\] Hence, for any
set
$\tilde{\Theta}_0$ of the form $\{ \theta_0(n_{at,1}), \ldots, \theta_0(n_{at,K}) \}$, we have that $$\tilde{\pi}_{A}(Y=y \mid \theta \in
\tilde{\Theta}_0) = \tilde{\pi}_A( Y=y \mid Y \in S_{\tilde{\Theta}_0}),$$ for
$S_{\tilde{\Theta}_0} = \bigcup_{\theta_0 \in \tilde\Theta_0} S^{Y}_{\theta_0}$.

Similarly, let $\theta_1(n_c) = (0,0,n-n_{c},n_{c})$, which has only compliers and defiers. The support points of $Y$ under
$\theta_1$ take the form 
\begin{equation}
    S^Y_{\theta_1(n_{c})} = \{(m, m + n_U - n_c) \,:\,  m=\max(0, n_c - n_U),\ldots,\min
(n_T,n_c) \} , \label{eqn:s1} 
\end{equation}
where $(m,n_U-(n_c-m)) = (m,m + n_U - n_c)$ corresponds to the realized outcomes when $m$ compliers are selected for treatment. These points lie
on an upward sloping line with slope 1 and intercept $n_U - n_c$. 
Letting $S^{Y}_{\theta_1}$ denote
the support of $Y$ under $\theta_1$, we again see that the $S^{Y}_{\theta_1}$ are
disjoint. Hence, for any set $\tilde{\Theta}_1$ of the form $\{ \theta_1(n_{c,1}), \ldots, \theta_1(n_{c,K}) \}$, we have that $$\tilde{\pi}_{B}(Y=y \mid
\theta \in \tilde{\Theta}_1) = \tilde{\pi}_B( Y=y \mid Y \in
S_{\tilde{\Theta}_1}),$$ for $S_{\tilde{\Theta}_1} = \bigcup_{\theta_1 \in \tilde\Theta_1}
S^{Y}_{\theta_1}$.

Now, let  
\begin{equation}
\tilde\Theta_0 \coloneqq \bigcup_{\substack{0\le n_{at}\le n\\ n_{at}\equiv n_T+1\, (\mathrm{mod}\ 2)}} \{\theta_0(n_{at})\} \subset \Theta_0    \label{eqn:defTildeTheta0}
\end{equation}
and
\begin{equation}
 \tilde\Theta_1 \coloneqq \bigcup_{\substack{1\le n_c\le n-1\\ n_c\ \mathrm{odd}}} \{\theta_1(n_c)\} \subset \Theta_1. \label{eqn:defTildeTheta1}   
\end{equation}
\Cref{lem:equivsupports} shows formally that \[
S_{\tilde{\Theta}_1} = S_{\tilde{\Theta}_0}. 
\]
\noindent An intuitive illustration of the proof is shown in \Cref{fig:supports-same} for the setting where $n_T=n_U=2$, in which case $S_{\tilde{\Theta}_1}$ corresponds to the union of the points on the two upward-sloping lines in orange, and $S_{\tilde{\Theta}_0}$ corresponds the unions of the points on the two downward-sloping lines in blue.

\begin{figure}[!ht]
    \centering
    \begin{tikzpicture}
    \foreach \x in {1,2,3} {
        \foreach \y in {1,2,3} {
            \fill (\x,\y) circle (2pt);
        }
    }

    \draw[->] (0,0) -- (4,0) node[anchor=north] {$Y_U$};
    \draw[->] (0,0) -- (0,4) node[anchor=east] {$Y_T$};

    \foreach \x in {0,1,2} {
        \node[anchor=north] at (\x+1,-0.1) {\x};
    }
    \foreach \y in {0,1,2} {
        \node[anchor=east] at (-0.1,\y+1) {\y};
    }

    \draw[blue, thick] (0.8,2.2) -- (2.2,0.8);
    \node[blue, anchor=north west] at (2.2,0.8) {$S^{Y}_{\theta_0(1)}$};

    \draw[blue, thick] (1.8,3.2) -- (3.2,1.8);

    \node[blue, anchor=north west] at (3.2,1.8) {$S^{Y}_{\theta_0(3)}$};

    \draw[orange, thick] (0.8,1.8) -- (2.2,3.2);

    \node[orange, anchor=south east] at (2.2,3.2) {$S^Y_{\theta_1(1)}$};

    \draw[orange, thick] (1.8,0.8) -- (3.2,2.2);

    \node[orange, anchor=south west] at (3.3,2.2) {$S^Y_{\theta_1(3)}$};
\end{tikzpicture}

    \caption{Illustration of the sets $S_{\tilde{\Theta}_0}$ and $S_{\tilde{\Theta}_1}$ with $n_T = n_U = 2$}
    \label{fig:supports-same}
\end{figure}
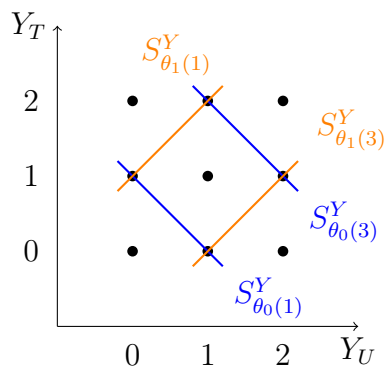

Define $\pi_A(\theta) := \tilde{\pi}_A(\theta \mid \theta \in \tilde{\Theta}_0)$ and
$\pi_B(\theta) := \tilde{\pi}_B(\theta \mid \theta \in \tilde{\Theta}_1)$. It is immediate
that $\pi_A(\theta \in \Theta_0) = 1$ and $\pi_B(\theta \in \Theta_1) = 1$ by
construction. Further, we have shown that \[
\pi_A(Y{=}y) = \tilde \pi_A(Y=y \mid Y \in S_{\tilde{\Theta}_0}) = \tilde \pi_B(Y=y \mid Y
\in S_{\tilde{\Theta}_1}) = \pi_B(Y{=}y) .
\qedhere
\]
\end{proof}

\begin{lem}
\label{lem:equivsupports}
Suppose $n$ is even. Let $\tilde{\Theta}_0$, $\tilde{\Theta}_1$ be as defined in \eqref{eqn:defTildeTheta0}-\eqref{eqn:defTildeTheta1}. Define
\[
C\coloneqq\bigl\{(a,b)\in\mathbb{Z}^2:\ 0\le a\le n_T,\ 0\le b\le n_U,\ a+b\equiv n_T+1\pmod 2\bigr\}
\]
to be the set of points in $\{0,\ldots,n_T\} \times \{0,\ldots,n_U\}$ whose sum has the opposite parity of $n_T$. Then
\[
S_{\tilde{\Theta}_0} = C =S_{\tilde{\Theta}_1}.
\]
\end{lem}

\begin{proof}
First, observe that since $n = n_U + n_T$ is even, we have that $n_U\equiv n_T\pmod 2$. 

\textbf{Step 1: $S_{\tilde{\Theta}_0} \subseteq C$.}
If $(a,b) \in S_{\tilde{\Theta}_0}$ then $(a,b)\in S^Y_{\theta_0(n_{at})}$ for some $n_{at}$ such that $n_{at}\equiv n_T+1\pmod 2$. It is immediate from the expression for $S^Y_{\theta_0(n_{at})}$ in \eqref{eqn:s0} that $a+b=n_{at} \equiv n_T+1 \pmod 2$. Moreover, by definition, $S^Y_{\theta_0(n_{at})}$ is the support of $Y$ under $\theta_0(n_{at})$ and thus $0 \leq a \leq n_T$ and $0 \leq b \leq n_U$. Thus, $(a,b)\in C$.

\medskip
\textbf{Step 2: $S_{\tilde{\Theta}_1} \subseteq C$.}
If $(a,b) \in S_{\tilde{\Theta}_1}$, then $(a,b)\in S^Y_{\theta_1(n_c)}$ for some odd $n_c$ such that $1 \leq n_c \leq n - 1$. It is immediate from the expression for $S^Y_{\theta_1(n_c)}$ in \eqref{eqn:s1} that $b-a=n_U-n_c$ and hence
\[
a+b\equiv b-a\equiv n_U-n_c \pmod 2,
\]
where the first $\equiv$ uses the fact that $2a \equiv 0 \pmod2$. Using the facts that $n_U\equiv n_T\pmod 2$ and $n_c$ is odd, we see that $n_U-n_c\equiv n_U-1\equiv n_T+1\pmod 2$. Further, by construction $S^Y_{\theta_1(n_c)}$ is the support of $Y$ under $\theta_1(n_c)$ and thus $0 \leq a \leq n_T$ and $0 \leq b \leq n_U$. Thus, $(a,b)\in C$.

\medskip
\textbf{Step 3: $C\subseteq S_{\tilde{\Theta}_0}$.}
Take $(a,b)\in C$ and set $n_{at}\coloneqq a+b$ and $m\coloneqq a$.
Then $n_{at}\equiv n_T+1\pmod 2$, and $0\le n_{at}\le n_T+n_U=n$.
We claim $(a,b)\in S^Y_{\theta_0(n_{at})}$.
Indeed, $b\le n_U$ implies $n_{at}-n_U=a+b-n_U\le a$, so $m=a\ge \max(0,n_{at}-n_U)$.
Also $a\le n_T$ and $a\le a+b = n_{at}$, so $m=a\le \min(n_T,n_{at})$.
Thus $(a,b)=(m,n_{at}-m)\in S^Y_{\theta_0(n_{at})}\subseteq S_{\tilde{\Theta}_0}$.

\medskip
\textbf{Step 4: $C\subseteq S_{\tilde{\Theta}_1}$.}
Take $(a,b)\in C$ and set
\[
 m\coloneqq a,\qquad n_c\coloneqq a+n_U-b.
\]
First, $0\le n_c\le n$ because $0\le a\le n_T$ and $0\le b\le n_U$ and $n_U + n_T =n$.
Second, $n_c$ is odd:
\[
 n_c =  a+n_U-b = (a-b)+n_U\equiv (a+b)+n_U\pmod 2,
\]
where the last $\equiv$ uses the fact that $2b \equiv 0 \pmod 2$. Further, since $(a,b)\in C$, $a+b\equiv n_T+1\pmod 2$, so
\[
 n_c\equiv (n_T+1)+n_U \equiv n+1 \equiv 1\pmod 2,
\]
where we use the fact that $n$ is even. Finally, we check that $m=a$ lies in the admissible range for $S^Y_{\theta_1(n_c)}$.
Because $n_c-n_U=a-b$, we have $\max(0,n_c-n_U)=\max(0,a-b)\le a=m$.
Also $m=a\le n_T$ and $m=a\le n_c$ (since $n_c=a+n_U-b\ge a$).
Therefore $(a,b)=(m, m+n_U-n_c)\in S^Y_{\theta_1(n_c)}\subseteq S_{\tilde{\Theta}_1}$.

Combining Steps 1--4 proves $S_{\tilde{\Theta}_0}=C=S_{\tilde{\Theta}_1}$.
\end{proof}

\prophasnopower*

\begin{proof}
Fix some $n$ that is even.
By \Cref{lem: there exist equiv priors on the null and alt}, there exist priors $\pi_A$ and $\pi_B$ such that $\pi_A(\theta \in \Theta_0) = 1$, $\pi_B(\theta \in \Theta_1) = 1$, and $\pi_A(Y =y) = \pi_B(Y =y)$ for all $y$. Observe that
$$E_{\theta \sim \pi_A}[ \delta(Y) ] = E_{\theta \sim \pi_A}[ \underbrace{E_\theta[ \delta(Y) ]}_{\text{$\leq \alpha$ since $\theta \in \Theta_0$}} ] \leq \alpha,$$
\noindent where the first equality uses the law of iterated expectations, and the second the assumption that $E_\theta[\delta(Y)] \leq \alpha$ for all $\theta \in \Theta_0$. However, since $\pi_A(Y =y) = \pi_B(Y =y)$ for all $y$, it follows that 
$$E_{\theta \sim \pi_B}[\delta(Y)] = E_{\theta \sim \pi_A}[\delta(Y)] \leq \alpha$$
and hence
$$E_{\theta \sim \pi_B}[\delta(Y)] = E_{\theta \sim \pi_B}[ E_\theta[\delta(Y)] ] \leq \alpha.$$
\noindent It follows that there exists some $\theta_1 \in \supp(\pi_B) \subseteq \Theta_1$ such that $E_{\theta}[\delta(Y)] \leq \alpha$, which gives the first desired result. 

For the second result, fix some odd $n$ and test $\delta$ such that $\sup_{\theta \in \Theta_0} E_{\theta}[\delta(Y)] \leq \alpha$. As in the proof to \Cref{lem: there exist equiv priors on the null and alt}, let $\tilde\pi_A$ and $\tilde\pi_B$ be the multinomial priors over $\theta$ with parameters $p=(0.5,0.5,0,0)$ and $p=(0,0,0.5,0.5)$, respectively. As argued in the proof to \Cref{lem: there exist equiv priors on the null and alt}, $\tilde{\pi}_A(Y{=}y) = \tilde{\pi}_B(Y{=}y)$ for all $y$. This implies that
\[E_{\theta \sim \tilde{\pi}_B}[ E_{\theta}[\delta(Y)] ] = E_{\theta \sim \tilde{\pi}_A}[ E_{\theta}[\delta(Y)] ] \leq \alpha, \numberthis \label{eqn:equalpowerunderpitildes} \]
where the inequality uses the fact that, by construction, $\tilde{\pi}_A(\theta \in \Theta_0) = 1$ and that $\delta$ controls size. Next, observe that $\tilde\pi_B(\theta \in \Theta_0) = \tilde\pi_B( \theta \in \{(0,0,n,0),(0,0,0,n)\} = 0.5^{n-1} =: \epsilon$. By iterated expectations, we then have
\[ E_{\theta \sim \tilde{\pi}_B}[ E_{\theta}[\delta(Y)] ] = \epsilon E_{\theta \sim \tilde{\pi}_B}[ E_{\theta}[\delta(Y)] \mid \theta \in \Theta_0 ] + (1-\epsilon) E_{\theta \sim \tilde{\pi}_B}[ E_{\theta}[\delta(Y)] \mid \theta \in \Theta_1 ]  \]
and hence 
\begin{align*}
   E_{\theta \sim \tilde{\pi}_B}[ E_{\theta}[\delta(Y)] \mid \theta \in \Theta_1 ] &= \frac{1}{1-\epsilon} E_{\theta \sim \tilde{\pi}_B}[ E_{\theta}[\delta(Y)] ] - \frac{\epsilon}{1-\epsilon} E_{\theta \sim \tilde{\pi}_B}[ E_{\theta}[\delta(Y)] \mid \theta \in \Theta_0 ] \\
   &\leq \frac{\alpha}{1-\epsilon} = \alpha \frac{1}{1 - 2^{-(n-1)}} = \alpha \left(1 + \frac{1}{2^{n-1}-1} \right)
\end{align*}
where the inequality uses \eqref{eqn:equalpowerunderpitildes} together with the fact that $\delta(Y) \geq 0$. The result then follows from the fact that $\inf_{\theta \in \Theta_1} E_{\theta}[\delta(Y)] \leq E_{\theta \sim \tilde{\pi}_B}[ E_{\theta}[\delta(Y)] \mid \theta \in \Theta_1 ]$.
 \end{proof}

\propnowap*
\begin{proof}
    Let $\pi(t; \theta_1)$ be the PMF for $\mathrm{Multinomial}(n, \theta_1/n)$. Then \[
        w_n(t;\theta_1) = \frac{\pi(t;\theta_1)}{\sum_{t\in\Theta_1} \pi(t;\theta_1)} \one
        (t\in\Theta_1).
    \]
    Let $q=\theta_1/n$ and observe that $g_{q}(y) = \sum_{t \in \Theta_0 \cup \Theta_1} f_t(y) \pi(t;\theta_1)$. By
    \Cref{prop:id}, there exists some $p\in \mathcal P_0$ such that $g_{q}(y) = g_p
    (y)$. We thus have the following expansion of WAP, \begin{align*}
    \mathrm{WAP}(\delta;\theta_1) &= \sum_{t \in \Theta_0 \cup \Theta_1} \sum_y \delta(y) f_t(y)\frac{\pi(t;\theta_1)}{\sum_{t\in\Theta_1} \pi(t;\theta_1)} \one
        (t\in\Theta_1) \\
        &= \underbrace{\pr{\sum_{t\in\Theta_1} \pi(t;\theta_1)}^{-1}}_{=: Z_1^{-1}}
        \sum_ {y} \delta (y) \bk{\sum_ {t \in
        \Theta_0
        \cup \Theta_1} f_t(y) \pi(t; \theta_1) - \sum_{t\in \Theta_0} f_t(y)\pi(t;\theta_1)}
        \\
        &= \frac{1}{Z_1} \bk{E_{g_{q}}[\delta(Y)] - \sum_{t\in\Theta_0} E_t[\delta(Y)] \pi
        (t;
        \theta_1)} \\
        &\le \frac{1}{Z_1} E_{g_p}[\delta(Y)] \le \frac{\alpha}{Z_1}. \tag{$\delta$ controls size}
    \end{align*}
    Note that \[
        Z_1 = 1-P_{V \sim \mathrm{Multinomial}(n, q)} (\min(V_d, V_c) = 0) = 1 -
        \bk{
            (1-q_d)^n + (1-q_c)^n - (1-q_d - q_c)^n
        }.
    \]
    It is easy to check that $Z_1$ is increasing in $q_c$ and $q_d$, thus it is lower
    bounded: \[
        Z_1 \ge 1-2(1-v)^n + (1-2v)^n,
    \]
    where recall $v = \tfrac{1}{n} \min(\theta_{1,d},\theta_{1,c}) = \min(q_d,q_c)$. This proves the first inequality. The second inequality follows from examining that
    $v \ge 1/n$ and that \[
    \frac{1}{1-2(1-1/n)^n + (1-2/n)^n} \le 2.51 \text{ for all $n > 3$. } \qedhere
    \]
\end{proof}

\propunbiasedexists*
\begin{proof}
Let $\delta_0(Y) = \alpha \cdot 1\{ (Y_T,Y_U) \not\in \{ (n_T,0), (0,n_U) \} \} $ be the test that rejects with probability $0$ if $Y$ is either $(n_T,0)$ or $(0,n_U)$, and rejects with probability $\alpha$ otherwise. Let $S^Y_\theta$ be the support of $Y$ under parameter $\theta$. Note that by construction, for any $\theta_0 = (n_{at},n_{nt}, n_d, n_c) \in \Theta_0$, we have that 
\begin{align*}
& E_{\theta_0 }[\delta_0(Y)]  < \alpha \text{ if } S^Y_{\theta_0} \cap  \{ (n_T,0), (0,n_U) \} \neq \emptyset, \\
& E_{\theta_0 }[\delta_0(Y)]  =  \alpha \text{ if } S^Y_{\theta_0} \cap  \{ (n_T,0), (0,n_U) \}  = \emptyset .
\end{align*}

Now, we claim that if $(n_T -1 , 1) \in S_{\theta_0}^Y$ for $\theta_0 \in \Theta_0$, then
\[
S^Y_{\theta_0} \cap  \{ (n_T,0), (0,n_U) \}  \neq \emptyset .
\]
To see this, suppose first that $\theta_0$ has $n_d = 0$. Then $y_i(1) \geq y_i(0)$ for all $i= 1,\ldots,n$. If there exists a treatment allocation $D$ such that $Y(D) = (n_T-1,1)$, then there exists one $i$ for whom $(Y_i,D_i) = (1,0)$. It follows that $y_i(0) = 1$ and hence $y_i(1)=1$. Likewise, there must be one $j$ for whom $(Y_j,D_j) = (0,1)$, which implies that $Y_j(1) = 0$ and hence $Y_j(0) = 0$. Letting $\tilde{D}$ be the treatment allocation that swaps the assignments of $i$ and $j$ and otherwise preserves the allocation of $D$, we see that
\[
Y(\tilde{D}) = Y(D) + (1,-1) = (n_T,0).
\]
We have thus shown that $(n_T,0) \in S^Y_{\theta_0}$.

Similarly, suppose that $\theta_0$ has $n_c = 0$. Then $y_i(1) \leq y_i(0)$ for $i=1,\ldots,n$. If there exists a treatment allocation $D$ such that $Y(D) = (n_T-1,1)$, then there exists a set $A \subset [n]$ of size $n_T-1$ such that all $i \in A$ have $(Y_i,D_i) = (1,1)$, which implies that $y_i(1) = 1$ and hence $y_i(0) = 1$. Likewise, there exists a set $B \subset [n]$ of size $n_U - 1 = n_T - 1$ such that for all $j \in B$, $(Y_j,D_j) = (0,0)$, which implies that $Y_j(0) = 0$ and hence $Y_j(1) = 0$. Letting $\tilde{D}$ be the treatment allocation that swaps the treatment assignments of units in $A$ and $B$, we see that
\[
Y(\tilde{D}) = Y(D) + (-(n_T-1),n_T-1) = (0,n_T) = (0,n_U),
\]
and hence $(0,n_U) \in S^Y_{\theta_0}$. This completes the proof that
\[
S^Y_{\theta_0} \cap  \{ (n_T,0), (0,n_U) \}  \neq \emptyset .
\]

Now, let $\delta_1(Y) = 1\{ (Y_T,Y_U) = (n_T-1,1)\}$ be the test that rejects if $Y$ is $(n_T-1,1)$. The argument above implies that for all $\theta_0 \in \Theta_0$,
\begin{align*}
& E_{\theta_0 }[\delta_0(Y)]  < \alpha \text{ if } E_{\theta_0}[\delta_1(Y)] > 0, \\
& E_{\theta_0 }[\delta_0(Y)] \leq  \alpha \text{ if } E_{\theta_0}[\delta_1(Y)] = 0 .
\end{align*}
Since $\Theta_0$ is finite, it follows that there exists $\epsilon >0$ such that, for
\[
\delta(Y) = \delta_0(Y) + \epsilon \delta_1(Y),
\]
we have $E_{\theta_0}[ \delta(Y) ] \leq \alpha$ for all $\theta_0 \in \Theta_0$. 

Next, we claim that $E_{\theta_1}[\delta(Y)] \geq \alpha$ for all $\theta_1 \in \Theta_1$. Since
\[
E_{\theta_1}[\delta(Y)] \geq E_{\theta_1}[ \delta_0(Y) ] 
= \alpha \, P_{\theta_1}\!\left( Y \not\in \{(n_T,0), (0,n_U) \} \right),
\]
it suffices to show that, for all $\theta_1 \in \Theta_1$,
\[
S_{\theta_1}^Y \cap \{(n_T,0), (0,n_U) \} = \emptyset,
\]
in which case $E_{\theta_1}[\delta_0(Y)] = \alpha$. To show this, we prove the contrapositive: if $S_{\theta}^Y$ contains $(n_T,0)$ or $(0,n_U)$, then $\theta \in \Theta_0$. Indeed, if there is a treatment allocation with $(Y_T,Y_U) = (n_T,0)$, then all $n_T$ treated units must be always-takers or compliers, and all $n_U$ control units must be never-takers or compliers, and thus there can be no defiers. Analogously, if there is a treatment allocation with $(Y_T,Y_U) = (0,n_U)$, then all treated units must be never-takers or defiers, and all control units must be always-takers or defiers, and thus there are no compliers.

Finally, we show that there exists $\theta_1$ such that $E_{\theta_1}[\delta(Y)] > \alpha$. Since we showed above that $E_{\theta_1}[\delta_0(Y)] = \alpha$ for all $\theta_1 \in \Theta_1$, it suffices to show that there exists some $\theta_1 \in \Theta_1$ such that $E_{\theta_1}[\delta_1(Y)] > 0$, or equivalently $P_{\theta_1}( Y = (n_T-1,1) ) > 0$. However, if $\theta_1$ corresponds to $(n_{at},n_{nt},n_d,n_c) = (0,0,2,n-2)$, then $Y = (n_T-1,1)$ obtains when one defier is assigned to treatment and the other to control.
\end{proof}

\propunbiased*
\begin{proof}
    The maximum power at $\theta$ for $\delta$ is defined by the following linear program:
    \begin{align*}
    E_\theta[\delta(Y)] \le \max_{\delta} \sum_{y} f_\theta(y) \delta(y) \text{ subject to } 
    &\sum_y f_t(y) \delta(y) \le \alpha \, \forall t \in \Theta_0, \\
&\sum_y f_t(y) \delta(y) \ge \alpha \, \forall t \in \Theta_1, \\
&\delta(y) \in [0,1].
    \end{align*}
    The dual program is \[
        \min_{\lambda(t)\ge 0, \mu(t)\ge 0} \alpha\bk{\sum_{t\in\Theta_0} \lambda(t) - \sum_
        {t\in\Theta_1} \mu(t)} + \sum_y \bk{
            f_{\theta}(y) - \sum_{t \in \Theta_0} \lambda(t) f_t(y) + \sum_{t\in \Theta_1}
            \mu(t) f_t(y)
        }_+.
    \]
    Thus, any dual feasible $\lambda, \mu$ implies an upper bound for $E_\theta[\delta(Y)].$
    
    Let $p_1 = \theta /n \in \mathcal P_1$ and let $\pi(t; p_1)$ be the PMF for $
    \mathrm{Multinomial}(n,p_1)$. Note that $\pi(\theta; p_1) > 0$ and define $c =
    \frac{1}{\pi(\theta; p_1)}$. Define, over all of $\Theta_0 \cup \Theta_1$, \[
        \mu(t) = \begin{cases}
            c \pi(t;p_1), &t \neq \theta \\
0 & t = \theta           
        \end{cases}.
    \]
    Then \[
       \sum_{t\in \Theta_1} \mu(t) f_t(y) + f_\theta(y) \le \sum_{t\in
       \Theta_1\cup\Theta_0} \mu(t) f_t(y) + f_\theta(y) = c g_{p_1}(y) = c g_
        {p_0}(y)
    \]
    for some $p_0 \in \mathcal P_0$ by \Cref{prop:id}. Define \[
        \lambda(t) = c\pi(t; p_0) \implies \sum_{t \in \Theta_0} \lambda(t) f_t(y) = c g_
        {p_0}(y).
    \]
    With this choice of $\lambda, \mu$, \[
        \bk{
            f_{\theta}(y) - \sum_{t \in \Theta_0} \lambda(t) f_t(y) + \sum_{t\in \Theta_1}
            \mu(t) f_t(y)
        }_+ = 0.
    \]

    Now, \[
        \sum_{t \in \Theta_0} \lambda(t) - \sum_{t\in \Theta_1} \mu(t) = c - \bk{c - 1 - c
        \sum_{t\in \Theta_0} \pi(t; p_1) } = 1 + \frac{\pi(\Theta_0; p_1)}{\pi(\theta; p_1)}.
    \]
    Thus \[
        E_\theta[\delta(Y)]  \le \alpha \pr{1 + \frac{\pi(\Theta_0; p_1)}{\pi(\theta; p_1)}}.
    \]
    By the argument in the proof of \Cref{prop:nowap}, \[
        \pi(\Theta_0; p_1) \le 2(1-v)^n - (1-2v)^n \le 2(1-\epsilon)^n.
    \]
    Meanwhile, \begin{align*}
    \pi(\theta; p_1) = \frac{n!}{\theta_{at}!\theta_{nt}!\theta_d!\theta_{c}!} p_{1,at}^
    {\theta_{at}} p_{1,nt}^{\theta_{nt}} p_{1,d}^{\theta_{d}} p_{1,c}^{\theta_{c}}
    \end{align*}
    Let $ m \in \br{2,3,4}$ be the number of entries in $\theta$ that are positive.
    Stirling's formula implies the bounds \citep{robbins1955remark}: \[
      \sqrt{2\pi} k^{k+1/2} e^{-k} \le k! \le e k^{k+1/2} e^{-k}
    \]
    for all integer $k \ge 1$. Plug in these bounds with $p_1 = \theta/n$ to obtain \[
          \pi(\theta; p_1) \ge \frac{\sqrt{2\pi}}{e^m} n^{-(m-1)/2} \prod_{j: \theta_j >
          0} p_{1j}^{-1/2} \ge m^{m/2} e^{-m} \sqrt{2\pi} n^{-(m-1)/2} \ge 0.64 n^{-1.5}
    \]
where the second inequality uses the AM-GM inequality\[
    \prod_{j: \theta_j >
          0} p_{1j} \le \pr{\frac{1}{m} \sum_j p_j}^m = m^{-m}.
\]
Therefore, \[
    E_\theta [\delta(Y)] \le \alpha(1+3.125 n^{1.5}(1-\epsilon)^n) \to \alpha
\]
as $n \to \infty$.
\end{proof}

\bayescombined*

\begin{proof}
    Let $\theta_0 \in \Theta_0$ and $\theta_1 \in \Theta_1$ be two parameter values. By
    \Cref{prop:id}, $f_{\theta_0} \neq f_{\theta_1}$. Let $\pi = 0.5 \delta_{\theta_0} +
    0.5 \delta_{\theta_1}$ be such that $\pi(\Theta_0) = 0.5 \in (0,1)$. The posterior
    probability is such that \[
        \pi(\Theta_0 \mid Y=y) = \pi(\theta_0 \mid Y=y) = 1-\pi(\theta_1 \mid Y=y) = \frac{f_{\theta_0}(y)}{f_
        {\theta_1}(y) + f_{\theta_0}(y)}.
    \]
    Since $f_{\theta_0} \neq f_{\theta_1}$ for some $y$ in the support of one of $f_
    {\theta_0}(\cdot)$ and $f_{\theta_1}(\cdot)$, $\pi(\theta_0 \mid Y=y) \neq 0.5 = \pi
    (\theta_0)$ with positive probability. This concludes the proof of the first part.
    
    For the second part, by \Cref{lem: there exist equiv priors on the null and alt}, there exists priors $\pi_A$ and $\pi_B$ such that $\pi_A(\theta \in \Theta_0) = 1$, $\pi_B(\theta \in \Theta_1)$, and $\pi_A(Y =y) = \pi_B(Y =y)$ for all $y$. Consider the prior $\pi_C = c \pi_A + (1-c) \pi_B$ which mixes with probabilities $c$ and $1-c$ between $\pi_A$ and $\pi_C$. By Bayes' rule, we have that 
$$\pi_C(\theta \in \Theta_0 \mid Y) = \dfrac{\pi_C(Y \mid \theta \in \Theta_0) \cdot  \pi_C(\theta \in \Theta_0)}{\pi_C(Y)} . $$
\noindent However, since $\pi_A(\theta \in \Theta_0) = 1$ and $\pi_B(\theta \in \Theta_0) = 0$, we have that $\pi_C(Y \mid \theta \in \Theta_0) = \pi_A(Y)$. Further $\pi_C(Y) = c \pi_A(Y) + (1-c) \pi_B(Y) = \pi_A(Y)$, where the second equality uses the fact that $\pi_A(Y) = \pi_B(Y)$. It follows that $\pi_C(Y \mid \theta \in \Theta_0) = \pi_A(Y) = \pi_C(Y)$. The previous display thus reduces to $\pi_C(\theta \in \Theta_0 \mid Y) = \pi_C(\theta \in \Theta_0) =c$, which gives the second result.

\end{proof}

\trivialclassification*
\begin{proof}
Observe that for any prior $\pi$ over $\theta$,
\begin{align*}
 \inf_{c(\cdot)} \sup_{\theta} R(c,\theta)
 &\geq \inf_{c(\cdot)} E_{\theta \sim \pi}[ R(c,\theta)] .
\end{align*}
For a (possibly randomized) classifier $c(\cdot)$, it follows that:
\begin{align*}
E_{\theta \sim \pi}[R(c,\theta)] & = E_{\theta \sim \pi}[ E_{Y \sim \pi_{Y \mid \theta}}[c(Y)\ell(1,\theta) + (1-c(Y) \ell(0,\theta)] ] \\
&= E_{Y \sim \pi_Y}[ E_{\theta \sim \pi_{\theta \mid Y}}[c(Y)\ell(1,\theta) + (1-c(Y) \ell(0,\theta)] ] \\
&= E_{Y \sim \pi_Y}\Big[
(1-\pi(\theta \in \Theta_0\mid Y))\,c(Y) + \pi(\theta \in \Theta_0\mid Y)\,(1-c(Y))
\Big],
\end{align*}
where for clarity in the previous display, we use e.g. $\pi_Y$ to denote the distribution of $Y$ induced by the prior $\pi$, and we make explicit the measures over which random variables are taken. 

By \Cref{prop:bayescombined} part 2, there exists a prior $\pi$ such that
$\pi(\theta \in \Theta_0 \mid Y) = \pi(\theta \in \Theta_0) = 0.5$ ($\pi$-a.s.), in which case we have
\begin{align*}
E_{\theta\sim\pi}[R(c,\theta)]
&= E_{Y \sim \pi_Y}\Big[ \tfrac12 c(Y) + \tfrac12 (1-c(Y))\Big]
= \tfrac12 .
\end{align*}
Hence for every classifier $c(\cdot)$, its Bayes risk under $\pi$ equals $1/2$, so in particular
\[
\inf_{c(\cdot)} E_{\theta\sim\pi}[R(c,\theta)] = \tfrac12,
\]
and therefore
\[
\inf_{c(\cdot)} \sup_{\theta} R(c,\theta) \ge \tfrac12.
\]
On the other hand, the trivial randomized classifier, $c(y)= 1/2$ for all $y$, has
\[
R(\theta,c)=\tfrac12 \quad \text{for all }\theta,
\]
and thus
\[
\inf_{c(\cdot)} \sup_{\theta} R(c,\theta) \le \tfrac12.
\]
Combining the two inequalities yields $\inf_{c(\cdot)} \sup_{\theta} R(c,\theta) = 1/2$.
\end{proof}

\begin{restatable}{prop}{bayesupdateodd}
\label{prop:bayesupdateodd}
For any $c \in (0,1)$, and all odd $n$ such that $0.5^{n-1} \leq c$, there exists a prior $\pi$ such that $\pi(\theta \in
\Theta_0) = c$ and
\[ E_{Y \sim \pi_Y}\left[ \left| \pi(\theta \in \Theta_0 \mid Y) - \pi(\theta \in \Theta_0)   \right| \right] = 2^{-(n-2)} (1-c) , \]
where $\pi_Y$ is the push-forward of $Y$ under $\theta \sim \pi$.
\end{restatable}

\begin{proof}
As in the proof to \Cref{lem: there exist equiv priors on the null and alt}, let
$\tilde\pi_A$ and $\tilde\pi_B$ be the multinomial priors over $\theta$ with parameters
$p=(0.5,0.5,0,0)$ and $p=(0,0,0.5,0.5)$, respectively. As argued in the proof to \Cref
{lem: there exist equiv priors on the null and alt}, $\tilde{\pi}_A(Y{=}y) = \tilde{\pi}_B
(Y{=}y)$ for all $y$. Observe that $\tilde{\pi}_A(\theta \in \Theta_0) = 1$ and hence
$\tilde\pi_A( \theta \in \Theta_0 \mid Y) = 1$. Next, observe that $\tilde\pi_B(\theta \in
\Theta_0) = \tilde\pi_B( \theta \in \{(0,0,n,0),(0,0,0,n)\} = 0.5^{n-1} =: \epsilon$. By assumption, $\epsilon \leq c$. Let $C := \{(n_T,0), (0,n_U)\}$
denote the corners of the support of $Y$, and observe that under $\tilde{\pi}_B$, $Y \in
C$ if and only if $\theta \in \Theta_0$: that is, $\tilde\pi_B( Y \in C \mid \theta \in
\Theta_0) = 1$, $\tilde\pi_B(Y \in C \mid \theta \in \Theta_1) = 0$.  It follows that
$\tilde\pi_B(Y \in C) = \epsilon$, and by Bayes' rule, $\tilde\pi_B(\theta \in \Theta_0
\mid Y) = 1\{ Y \in C \}$.

Now, let $\pi(\theta) = \omega \tilde\pi_A(\theta) + (1-\omega) \tilde\pi_B(\theta)$ be the mixture of $\tilde\pi_A$ and $\tilde\pi_B$, with the mixture weight $\omega := \tfrac{c-\epsilon}{1-\epsilon}$ chosen so that $\pi(\theta \in \Theta_0) = c$. Observe that
\begin{align*}
   \pi(\theta \in \Theta_0 \mid Y) &= \frac{ \pi(\theta \in \Theta_0, Y) }{ \pi(Y) } \\ 
   &= \frac{ \omega \tilde\pi_A(\theta \in \Theta_0 \mid Y) \tilde\pi_A(Y) + (1-\omega) \tilde\pi_B(\theta \in \Theta_0 \mid Y) \tilde\pi_B(Y)  }{ \pi(Y) }  \\
   &= \omega \tilde \pi_A(\theta \in \Theta_0 \mid Y) + (1-\omega) \tilde\pi_B(\theta \in \Theta_0 \mid Y) 
\end{align*}
where the first line uses the definition of conditional probability, the second line uses the definition of $\pi$, and the third line uses the fact that $\tilde\pi_A(Y) = \tilde\pi_B(Y) = \pi(Y)$. It follows from the previous display and our earlier derivations of the posterior probabilities that 
\[ \pi(\theta \in \Theta_0 \mid Y) = \omega + (1-\omega) 1\{Y \in C\}. \]
We thus see that
\begin{align*}
 E_{Y \sim \pi_Y}\left[ \left|  \pi(\theta \in \Theta_0 \mid Y) - \pi(\theta \in \Theta_0)   \right| \right] &= \pi(Y \in C) (1-c) + (1- \pi(Y\in C)) (c- \omega ) \\
 &= \epsilon \left(1 - c\right) + (1- \epsilon) (c - \omega) = 2 \epsilon (1-c).
\end{align*}
\noindent The result then follows from the definition of $\epsilon = 2^{-(n-1)}$.
\end{proof}

\begin{restatable}{cor}{classificationodd}
\label{cor:classification-odd}
Suppose $n \geq 3$ is odd. Then $\inf_{c(\cdot)} \sup_{\theta} R(c,\theta) \geq 0.5 - 2^{-(n-1)}$.
\end{restatable}

\begin{proof}
As argued in the proof to \Cref{cor:trivialclassification}, for any prior $\pi$,  $\inf_{c(\cdot)} \sup_{\theta} R(c,\theta) \geq \inf_{c(\cdot)} E_{\theta \sim \pi}[ R(c,\theta)]$ and 
\[E_{\theta \sim \pi}[R(c,\theta)] = E_{Y \sim \pi_Y}\Big[
(1-\pi(\theta \in \Theta_0\mid Y))\,c(Y) + \pi(\theta \in \Theta_0\mid Y)\,(1-c(Y))\Big]. \numberthis \label{eqn:risk-expression-odd-n} \]
From \Cref{prop:bayesupdateodd}, there exists $\pi$ such that $\pi(\theta \in \Theta_0) = 0.5$ and $E_{Y \sim \pi_Y}[ | \pi(\theta \in \Theta_0 \mid Y) - \pi(\theta \in \Theta_0) | ] = 2^{-(n-1)}$. Adding and subtracting $E_{Y \sim \pi_Y}[(1-\pi(\Theta_0)) c(Y) + \pi(\Theta_0) (1-c(Y))] = 0.5 $ from \eqref{eqn:risk-expression-odd-n}, we see that
\begin{align*}
E_{\theta \sim \pi}[R(c,\theta)]  &= 0.5 + E_{Y \sim \pi_Y}\Big[
(\pi(\Theta_0\mid Y) - \pi(\Theta_0) ) (1-2 c(Y)) \Big] \\
& \geq 0.5 - E_{Y \sim \pi_Y}[ | \pi(\Theta_0\mid Y) - \pi(\Theta_0) | \cdot |1-2c(Y)| ] \\
& \geq 0.5 - E_{Y \sim \pi_Y}[ | \pi(\Theta_0\mid Y) - \pi(\Theta_0) | ] = 0.5 - 2^{-(n-1)}
\end{align*}
where the second line uses the triangle inequality, and the third uses the fact that $c(Y) \in [0,1]$. This completes the proof.
\end{proof}

\estimation*

\begin{proof}

    We prove both parts by constructing a prior $\pi$ such that the Bayes estimators for $\theta$ and $v$ under $\pi$ have average risk bounded from below as in the statement of the proposition.

    We first construct a distribution $\pi^*$ that implies a prior $\pi$ that mixes over $\tilde{\pi}_A$ and $\tilde{\pi}_B$ from \Cref{lem: there exist equiv priors on the null and alt}.
    As in the proof of \Cref{lem: there exist equiv priors on the null and alt}, let $\tilde{\pi}_A$ and $\tilde{\pi}_B$ be the multinomial priors over $\theta$ with parameters $p = (0.5,0.5,0,0)$ and $p = (0,0,0.5,0.5)$, respectively.
    Let $U \sim \text{Uniform}\{A,B\}$ select between the two priors uniformly at random, and write $\pi^*$ for the implied distribution over $(U,\theta,Y)$ induced by $\theta \mid U \sim \tilde{\pi}_U$ and $Y \mid \theta \sim f_\theta$.
    Then the marginal distribution $\pi$ over $\theta$ implied by $\pi^*$ is the equal mixture $\pi(\theta) = .5 \: \tilde{\pi}_A(\theta) + .5 \: \tilde{\pi}_B(\theta)$ of $\tilde{\pi}_A$ and $\tilde{\pi}_B$ as in the proof of \Cref{prop:bayesupdateodd} with mixing weight $\omega = .5$.

    We next consider Bayes estimators under $\pi$.
    The Bayes estimators for $\theta$ and $v$ that minimize
    \begin{align*} 
        &E_{\theta \sim \pi}[ E_\theta [\| \hat{\theta} - \theta \|^2]]
        &
        &\text{ and }
        &
        &E_{\theta \sim \pi}[ E_\theta [| \hat{v} - v |^2 ]]
    \end{align*}
    are given by
    \begin{align*}
        \hat{\theta}_\pi(Y) &= E_{(U,\theta,Y) \sim \pi^*}[\theta \mid Y]
        &
        &\text{ and }
        &
        \hat{v}_\pi(Y) &= E_{(U,\theta,Y) \sim \pi^*}[v \mid Y].
    \end{align*}
    Since the prior $\tilde{\pi}_A$ is supported on always-takers and never-takers only, and the prior $\tilde{\pi}_B$ is supported on compliers and defiers only, 
    \begin{align*}
        E_{(U,\theta,Y) \sim \pi^*}[\theta \mid Y, U]
        &=
        \begin{cases}
            (Y_T + Y_U, n - Y_T - Y_U, 0, 0)
            , & U = A, \\
            (0,0,n_T - Y_T + Y_U, Y_T + n_U - Y_U)
            , & U = B,
        \end{cases}
        \\
        E_{(U,\theta,Y) \sim \pi^*}[v \mid Y, U]
        &=
        \begin{cases}
           0
            , & U = A, \\
            \tfrac{1}{n} \min(Y_T + n_U - Y_U, n_T - Y_T + Y_U)
            , & U = B.
        \end{cases}
    \end{align*}
    In particular,
    observing the data $Y$ and
    knowing which distribution $U$ the data comes from allows us to perfectly recover $\theta$ and thus also $v$,
    so
    \begin{align}
    \label{eqn:conditional-target}
        E_{(U,\theta,Y) \sim \pi^*}[\theta \mid Y, U]
        &= \theta,
        &
        E_{(U,\theta,Y) \sim \pi^*}[v \mid Y, U]
        &= v
    \end{align}
    almost surely under $\pi^*$.
    At the same time,
    the distribution of $Y$ is the same under $\tilde{\pi}_A$ and $\tilde{\pi}_B$ by the proof of \Cref{lem: there exist equiv priors on the null and alt}, so observing $Y$ does not update about $U$,
    \begin{align*}
        P_{(U,\theta,Y) \sim \pi^*}(U = A \mid Y)
        = P_{(U,\theta,Y) \sim \pi^*}(U = A)
        = 0.5.
    \end{align*}
    As a consequence,
    \begin{equation}
    \label{eqn:conditional-estimator}
        \begin{aligned}
        \hat{\theta}_\pi(Y) &= E_{(U,\theta,Y) \sim \pi^*}[E_{(U,\theta,Y) \sim \pi^*}[\theta \mid Y, U] \mid Y]
        \\
        &= .5 \: E_{(U,\theta,Y) \sim \pi^*}[\theta \mid Y, U = A] + .5 \: E_{(U,\theta,Y) \sim \pi^*}[\theta \mid Y, U = B]
        \\
        &= .5 \cdot (Y_T + Y_U, n - Y_T - Y_U, n_T - Y_T + Y_U, Y_T + n_U - Y_U),
        \\
        \hat{v}_\pi(Y) &= E_{(U,\theta,Y) \sim \pi^*}[E_{(U,\theta,Y) \sim \pi^*}[v \mid Y, U] \mid Y]
        \\
        &= .5 \: E_{(U,\theta,Y) \sim \pi^*}[v \mid Y, U = A] + .5 \: E_{(U,\theta,Y) \sim \pi^*}[v \mid Y, U = B]
        \\
        &= \frac{1}{2 n} \min(Y_T + n_U - Y_U, n_T - Y_T + Y_U).
    \end{aligned}
    \end{equation}

    Having solved for the Bayes estimators, we can compute average risk.
    From \eqref{eqn:conditional-target} and \eqref{eqn:conditional-estimator}, we have that
    \begin{align*}
        \hat{\theta}_\pi(Y) - \theta
        &=  \frac{(-1)^{1\{U=B\}}}{2} (-(Y_T + Y_U), -(n - Y_T - Y_U), +(n_T - Y_T + Y_U), +(Y_T + n_U - Y_U)),
        \\
        \hat{v}_\pi(Y) - v
        &= \frac{(-1)^{1\{U=B\}}}{2n}  \min(Y_T + n_U - Y_U, n_T - Y_T + Y_U)
    \end{align*}
    and thus
    \begin{align*}
        \| \hat{\theta}_\pi(Y) - \theta \|^2
        &= \frac{1}{4} \left( (Y_T + Y_U)^2 + (n - Y_T - Y_U)^2 + (n_T - Y_T + Y_U)^2 + (Y_T + n_U - Y_U)^2 \right)
        \\
        | \hat{v}_\pi(Y) - v |^2
        &= \frac{1}{4 n^2}  \min(Y_T + n_U - Y_U, n_T - Y_T + Y_U)^2.
    \end{align*}
    By symmetry under $\pi$, it follows for the average risk under $\pi$ that
    \begin{align*}
        E_{\theta \sim \pi}[ E_\theta [\| \hat{\theta}_\pi - \theta \|^2]]
        &= E_{Y \sim \pi_Y}[(Y_T + Y_U)^2]
        = E[Z^2],
        \\
        E_{\theta \sim \pi} [E_\theta [| \hat{v}_\pi - v |^2]]
        &= \frac{1}{4 n^2} E_{Y \sim \pi_Y}[\min(Y_T {+} n_U {-} Y_U, n_T {-} Y_T {+} Y_U)^2]
        = \frac{1}{4 n^2} E[\min(Z, n {-} Z)^2]
    \end{align*}
    for $Z \sim \text{Binomial}(n,0.5)$.

    For $\theta$, we obtain the average risk
    \begin{align*}
        E_{\theta \sim \pi} \bigg[E_\theta \bigg[\bigg\| \tfrac{1}{n} (\hat{\theta}_\pi - \theta) \bigg\|^2\bigg]\bigg]
        = \frac{1}{n^2} E[Z^2]
        = \frac{1}{4} + \frac{1}{4 n},
    \end{align*}
    which yields the second part of the proposition (where we use that minimax risk is bounded from below by maximal Bayes risk).

    For $v$, using the fact that for $a,b \geq 0$, $\min(a,b) = 0.5 (|a+b| - |a-b|)$, we have that    
    \begin{align*}
        &E[\min(Z, n - Z)^2]
        =
        \frac{1}{4}
        E \left[\left(|Z + (n - Z)| - |Z - (n - Z)| \right)^2 \right]
        \\
        &= n^2 / 4  - n E [|Z - E[Z]|] + E[(Z - E[Z])^2]
        =
        n^2 / 4 + n / 4 - n E |Z - E[Z]|.
    \end{align*}
    As shown in \citet{berend2013sharp} (using a formula from De Moivre), for
    $V \sim \text{Binomial}(n,q)$,
    \begin{align*}
        E\big[\left|V-E[V]\right|\big]
        &=2(1-q)^{\,n-\lfloor nq\rfloor}\,q^{\lfloor nq\rfloor+1}\,(\lfloor nq\rfloor+1)
        \binom{n}{\lfloor nq\rfloor+1}
        \;\le\; \sqrt{nq(1-q)}.
    \end{align*}
    Hence, taking $q=0.5$,
    \begin{align*}
        E\big[\bigl|Z - E[Z]\bigr|\big]
        &=
        \frac{\lfloor n/2\rfloor+1}{2^{n}} \binom{n}{\lfloor n/2\rfloor+1}
        \le
        \frac{\sqrt{n}}{2}.
    \end{align*}
    For the estimation of $v$, we obtain the Bayes risk (and thus lower bound on minimax risk)
    \begin{align*}
        E_{\theta \sim \pi}\big[ E_\theta \big[ | \hat{v}_\pi - v |^2 \big]\big]
        &=
        \frac{1}{16}
        + \frac{1}{16 n}
        - \frac{\lfloor n/2\rfloor+1}{4 n 2^{n}} \binom{n}{\lfloor n/2\rfloor+1}
        \\
        &\geq
        \frac{1}{16} \left(1 - 2 n^{-1/2} + n^{-1}\right)
        = \frac{1}{16} \left(1 - \frac{1}{\sqrt{n}}\right)^2.
    \end{align*}
    The equality in the first line of the display implies that for $n\ge 30$, minimax MSE is at least $74\%$ that of the
    trivial estimator that always guesses $\tfrac{1}{4}$, and for $n\ge 100$, it is at least $85\%$.
\end{proof}

\mleconvergencesimpler*

\begin{proof}
    The conditions imply those of \Cref{prop:mlefullconvergence}, so
    \[
        \tfrac{1}{n}\hat{\theta}^{\text{MLE}} \xrightarrow{p} p^* = \argmax_{p' \in \mathcal{P}(q_T,q_U)}\, \max(p'_{at}, p'_{nt}, p'_d, p'_c),
    \]
    with $\min(p^*_{at}, p^*_{nt}, p^*_d, p^*_c) = 0$,
    and it remains to identify $p^*$.
    Writing $v = p'_d$, the elements of $\mathcal{P}(q_T,q_U)$ are
    \[
        p'(v) = (q_U - v,\; 1 - q_T - v,\; v,\; q_T - q_U + v),
        \qquad
        v \in [0, v_{max}],
    \]
    with endpoints $p'(0) = p^L$ and $p'(v_{max}) = p^U$.
    As the maximum of a strictly decreasing and a strictly increasing affine function of $v$, $\max(p'_{at}(v), p'_{nt}(v), p'_d(v), p'_c(v))$ is uniquely maximized over $[0,v_{max}]$ at an endpoint,
    with values $m^L = \max(q_U,\, 1 - q_T,\, q_T - q_U)$ for $v = 0$ (in which case $p^* = p^L$) or $m^U = \max(\abs{q_T + q_U - 1},\, \min(q_T,\, 1 - q_U))$ for $v = v_{max}$ (in which case $p^* = p^U$).

    It remains to compare $m^L$ and $m^U$.
    If $q_T, q_U < \tfrac{1}{2}$, then $q_T + q_U < 1$, so $m^U = \max(1 - q_T - q_U,\, q_T)$, where both arguments are smaller than $1 - q_T \leq m^L$ (by $q_U > 0$ and $q_T < \tfrac{1}{2}$, respectively); hence $m^U < m^L$ and $p^* = p^L$.
    If $q_T, q_U > \tfrac{1}{2}$, then $q_T + q_U > 1$, so $m^U = \max(q_T + q_U - 1,\, 1 - q_U)$, where both arguments are smaller than $q_U \leq m^L$ (by $q_T < 1$ and $q_U > \tfrac{1}{2}$, respectively); hence again $m^U < m^L$ and $p^* = p^L$.
    If instead $q_U < \tfrac{1}{2} < q_T$, then $\min(q_T,\, 1 - q_U) > \tfrac{1}{2}$, while all three arguments of $m^L$ are smaller: $q_U, 1 - q_T < \tfrac{1}{2}$, and $q_T - q_U < \min(q_T,\, 1 - q_U)$ by $q_U > 0$ and $q_T < 1$; hence $m^L < m^U$ and $p^* = p^U$.
\end{proof}

\mleviolation*

\begin{proof}
    By \Cref{prop:mleconvergencesimpler} and continuous mapping,
    \[
        \hat{v}^{\text{MLE}} = \min\bigl(\tfrac{1}{n}\hat{\theta}^{\text{MLE}}_d,\, \tfrac{1}{n}\hat{\theta}^{\text{MLE}}_c\bigr)
        \xrightarrow{p}
        \min(p^*_d, p^*_c).
    \]
    If $q_T$ and $q_U$ lie on the same side of $\tfrac{1}{2}$, then $p^* = p^L$, and $\min(p^L_d, p^L_c) = \min(0,\, q_T - q_U) = 0$. If instead $q_U < \tfrac{1}{2} < q_T$, then $p^* = p^U$, and
    \[
        \min(p^U_d, p^U_c) = \min(v_{max},\, q_T - q_U + v_{max}) = v_{max} > 0,
    \]
    where we use $q_T \geq q_U$, and positivity follows from $0 < q_U, q_T < 1$.
\end{proof}

\newpage

\section{Supplemental appendix (limit experiment)}

In this section, we consider a normal limit experiment for the completely randomized experiment from \Cref{sec:setup} and then leverage the resulting approximation to derive asymptotic properties of the maximum-likelihood estimator $\hat{\theta}$.

We first formalize an appropriate limit experiment. 
Fix $\eta > 0$ and restrict the parameter space to those type counts for which at least three of the four types each make up at least a fraction $\eta$ of the population,
\[
    \Theta_{3,n}(\eta) = \br{\theta \in \Theta : \theta_j \geq \eta n \text{ for at least three of the four types } j \in \br{at, nt, d, c}}.
\]
For $\theta \in \Theta_{3,n}(\eta)$, let $\varphi_\theta$ denote the density of the bivariate normal distribution $\mathcal{N}\pr{\mu_\theta,\, V_\theta}$, where $V_\theta := \tfrac{n-1}{n} \, \Sigma_\theta$ with $\mu_\theta$ and $\Sigma_\theta$ as in \eqref{eqn:Normalapprox}. Let
\[
    \gamma_\theta(y) = \frac{\varphi_\theta(y)}{\sum_{z \in \mathbb{Z}^2} \varphi_\theta(z)}, \qquad y \in \mathbb{Z}^2,
\]
denote the probability mass function of its discretization to the integer lattice. We compare the experiment that observes a single draw of $(Y_T, Y_U)$ from the exact randomization distribution to the experiment that observes a single draw from the discretized normal approximation, both defined on the sample space $\mathbb{Z}^2$ with common parameter space $\Theta_{3,n}(\eta)$:
\[
    \mathcal{E}_n = \br{f_\theta : \theta \in \Theta_{3,n}(\eta)},
    \qquad
    \mathcal{G}_n = \br{\gamma_\theta : \theta \in \Theta_{3,n}(\eta)}.
\]

Having defined an appropriate limit experiment, we can now state the sense in which it approximates the experiment from \Cref{sec:setup} in large samples:

\begin{prop}
\label{prop:lecam}
\emph{(Asymptotic Le Cam equivalence)}
Fix $\lambda \in (0, \tfrac{1}{2})$ and $\eta > 0$, and consider $n \rightarrow \infty$ with $\tfrac{n_T}{n} \in [\lambda, 1 - \lambda]$. Then
\(
    \Delta(\mathcal{E}_n, \mathcal{G}_n) \rightarrow 0.
\)
\end{prop}

Here, for two experiments $\mathcal{E} = \br{P_\theta : \theta \in \Theta'}$ and $\mathcal{F} = \br{Q_\theta : \theta \in \Theta'}$ with common parameter space $\Theta'$, we define the Le Cam deficiency and distance as, respectively,
\[
    \delta(\mathcal{E}, \mathcal{F}) = \inf_K \sup_{\theta \in \Theta'} \norm{K P_\theta - Q_\theta}_{\mathrm{TV}},
    \qquad
    \Delta(\mathcal{E}, \mathcal{F}) = \max\pr{\delta(\mathcal{E}, \mathcal{F}), \delta(\mathcal{F}, \mathcal{E})},
\]
where the infimum is taken over Markov kernels $K$ between the two sample spaces.

Based on this asymptotic equivalence, we derive the large-sample probability limit of the maximum-likelihood estimator:

\begin{prop}
\label{prop:mlefullconvergence}
\emph{(MLE convergence)}
Consider a sequence of type counts $\theta_n$ with $\tfrac{1}{n} \theta_n \rightarrow p \in [0,1]^4$ along the sequence. Assume that $n_T/n \in [\lambda, 1-\lambda]$ for some fixed $\lambda \in (0, 1/2)$. 
Let $q_T = p_{at} + p_c, q_U = p_{at} + p_d$ and assume that $q_T,q_U \in (0,1) \setminus \{\tfrac{1}{2}\}$ with $q_T \neq q_U$ and $q_T + q_U \neq 1$.
Then
\begin{align*}
    \tfrac{1}{n}\hat{\theta}^{\text{MLE}} &\xrightarrow{p} p^* = p^*(q_T,q_U),
    &
    \min(p^*_{at}, p^*_{nt}, p^*_d, p^*_c)
    &= 0
\end{align*}
as $n \rightarrow \infty$,
where $p^*(q_T,q_U)$ only depends on $p$ through $q_T,q_U$ and has exactly three non-zero elements (so that the proportion of one type is zero).
Furthermore,
\begin{align*}
p^*(q_T,q_U) = \argmax_{p' \in \mathcal{P}(q_T,q_U)}\, \max(p'_{at}, p'_{nt}, p'_d, p'_c)
\end{align*}
where
\(
\mathcal{P}(q_T,q_U) = \{p' \in [0,1]^4 : p'_{at} {+} p'_{nt} {+} p'_d {+} p'_c {=} 1,\; p'_{at} {+} p'_c {=} q_T,\; p'_{at} {+} p'_d {=} q_U\}
\)
is the set of type fractions consistent with the marginal fractions $q_T,q_U$. 
\end{prop}

The proofs of these results are based on a uniform central limit theorem for Bernoulli sums.
Below, we state the main local central limit theorem (\Cref{prop:lattice}) in \Cref{sec:localclt}, as well as multiple lemmas leveraged in its proof in the following subsections. We then put all these pieces together in \Cref{sec:lecamconv,sec:mleconv} to prove \Cref{prop:lecam,prop:mlefullconvergence}.

\subsection{Local CLT for Bernoulli sums}
\label{sec:localclt}

\begin{prop}[Uniform local estimate for Bernoulli direction sums]\label{prop:lattice}
Fix \(\lambda\in(0,1/2)\) and \(\eta>0\).  Let
\(B_{n1},\ldots,B_{nn}\) be independent Bernoulli random variables with common
success probability \(p_n\in[\lambda,1-\lambda]\), and let
\[v_{n1},\ldots,v_{nn}\in\mathcal V = \br{(1,1,1), (1,0,0), (1,0,1), (1,1,0)}\]
be deterministic. Define
\[
X_{ni}=B_{ni}v_{ni},\qquad
S_n=\sum_{i=1}^nX_{ni},\qquad
\Lambda_n(h)=\log E\exp\{h'S_n\},
\]
and
\[
\mu_n=E[S_n]=\nabla\Lambda_n(0),\qquad
\Gamma_n = \Var(S_n) =\nabla^2\Lambda_n(0).
\]
Assume that at least three of the four counts
\[
N_n(v)=\#\{i:v_{ni}=v\},\qquad v\in\mathcal V,
\]
are at least \(\eta n\).  Then there is a deterministic sequence
\(\varepsilon_n=\varepsilon_n(\lambda,\eta)\to0\) such that
\[
P(S_n=z)
\le
\frac{1+\varepsilon_n}{(2\pi)^{3/2}\sqrt{\det\Gamma_n}}
\qquad\text{for every }z\in\mathbb Z^3.
\]
Moreover, for every fixed \(\delta\in(0,1/6)\), after possibly replacing
\(\varepsilon_n\) by a sequence depending on \((\lambda,\eta,\delta)\),
\[
P(S_n=z)
=
\frac{1+R_{n,z}}{(2\pi)^{3/2}\sqrt{\det\Gamma_n}}
\exp\left\{-\frac12(z-\mu_n)'\Gamma_n^{-1}(z-\mu_n)\right\} \numberthis
\label{eq:localexpansion}
\]
uniformly over all \(z\) satisfying
\[
P(S_n=z) > 0 \text{ and } \|z-\mu_n\|\le n^{1/2+\delta},
\]
with \(\sup_z |R_{n,z}|\le\varepsilon_n\).
\end{prop}

\begin{proof}
For $h\in \mathbb R$, define
\[
K_n(h,u)
=
\exp\{-\Lambda_n(h)-iu'\nabla\Lambda_n(h)\}
E\exp\{(h+iu)'S_n\}.
\]
Also write
\[
\pi_{ni}(h):=\frac{p_ne^{h'v_{ni}}}{1-p_n+p_ne^{h'v_{ni}}},
\qquad
\Gamma_n(h):=\nabla^2\Lambda_n(h).
\]
Fourier inversion applied to the Fourier transform of the weighted mass function
\(s\mapsto e^{h's}P(S_n=s)\) gives
\begin{align*}
P(S_n=z)
&=
e^{-h'z}(2\pi)^{-3}
\int_{[-\pi,\pi]^3}
e^{-iu'z}E\exp\{(h+iu)'S_n\}\,du \numberthis \label{eq:weighted-inversion} \\
&=
e^{\Lambda_n(h)-h'z}(2\pi)^{-3}
\int_{[-\pi,\pi]^3}
e^{-iu'\{z-\nabla\Lambda_n(h)\}}K_n(h,u)\,du
\end{align*}
for every real $h$.

We first prove \eqref{eq:localexpansion}.  Fix a support point \(z\) satisfying
\(\|z-\mu_n\|\le n^{1/2+\delta}\).  By \Cref{lem:saddlepoint}, there exists some $C > 0$ and
\(h=h(z)\) such that
\(
\nabla\Lambda_n(h)=z, \text{ with  } \|h\|\le Cn^{-1/2+\delta}.
\)
Plugging in this choice of \(h\), we can simplify:
\[
P(S_n=z)
=
e^{\Lambda_n(h)-h'z}(2\pi)^{-3}
\int_{[-\pi,\pi]^3}K_n(h,u)\,du .
\]

Fixing this $h$, we next evaluate \(\Lambda_n(h)-h'z\).  By \Cref{lem:saddlepoint},
uniformly
over the present class of arrays and support points, we have a Taylor approximation
$\Lambda_n(h) - h'z \approx -\frac{1}{2}(z-\mu_n)'\Gamma_n^{-1} (z-\mu_n)$:
\(
\left|
\Lambda_n(h)-h'z
+\frac12(z-\mu_n)'\Gamma_n^{-1}(z-\mu_n)
\right|
\le Cn\|h\|^3.
\)
Since \(\|h\|\le Cn^{-1/2+\delta}, \delta < 1/6\), we thus have
\[
\exp\pr{\Lambda_n(h) - h'z} = (1+o(1))\exp\pr{-\frac{1}{2} (z-\mu_n)' \Gamma_n^{-1}
(z-\mu_n)}. \numberthis
\label{eq:expansion_Lambda}
\]

Next, we consider $\int K_n(h,u)\,du$.  \Cref{lem:K-small-frequency} gives the
approximation \[
	K_n(h,u) = \exp\pr{-\frac{1}{2} u'\Gamma_n(h) u + O(n\norm{u}^3)}, \numberthis
	\label{eq:K-approx}
\]
for small $u$ and $h$,
which motivates the following decomposition into $\norm{u} \le \rho_0$ and $\norm{u} \ge
\rho_0$:
\begin{align*}
\int_{[-\pi, \pi]^3} K_n(h,u)du
& = \mathrm{I}_n(h)+\mathrm{II}_n(h),\\
\mathrm{I}_n(h)
&:=\int_{\{\norm{u} > \rho_0\}\cap[-\pi,\pi]^3}  K_n(h,u)du,\\
\mathrm{II}_n(h)
&:=\int_{\{\norm{u} \le \rho_0\}\cap[-\pi,\pi]^3}  K_n(h,u)du.
\end{align*}

Let us first handle term $\mathrm{II}_n(h)$ using \eqref{eq:K-approx}. Set $R_n = n^{1/12}$.
By changing variables to
\(v := \Gamma_n(h)^{1/2} u\),
we have
\begin{align*}
\mathrm{II}_n(h)&=\frac{1}{\sqrt{\det\Gamma_n(h)}}\{\mathrm{III}_n(h)+\mathrm{IV}_n(h)\},\\
\mathrm{III}_n(h)&:=\int_{\{\|v\|\le R_n\}}K_n(h,\Gamma_n(h)^{-1/2}v)\,dv \\
\mathrm{IV}_n(h)&:=\int_{A_{n,h}\cap\{\|v\|>R_n\}}
K_n(h,\Gamma_n(h)^{-1/2}v)\,dv,
\end{align*}
for $A_{n,h}:=\{v:\|\Gamma_n(h)^{-1/2}v\|\le \rho_0\}$ and sufficiently large $n$,
where we assume without loss of generality that $\rho_0 \leq \pi$.
In terms of $u$, since $\Gamma_n
= \Theta(n)$, we are truncating on whether $\norm{u} = O(n^{-1/2} R_n) = O(n^{-5/12})$.

For \(\mathrm{III}_n(h)\), since \(\norm{u} = \|\Gamma_n(h)^{-1/2}v\|\le Cn^{-1/2}R_n\) on
\(\|v\|\le R_n\),
\Cref{lem:K-small-frequency} gives
\[
K_n(h,\Gamma_n(h)^{-1/2}v)
=e^{-\|v\|^2/2}\exp\{\widetilde \xi_n(h,v)\}, \quad \tilde \xi_n(h, v) = \xi_n(h, \Gamma_n(h)^
{-1/2} v)
\]
for sufficiently large $n$,
where
$
|\widetilde \xi_n(h,v)|\le Cn^{-1/2}R_n^3=Cn^{-1/4}.
$
Therefore
\[
\mathrm{III}_n(h)
= (1+O(n^{-1/4})) \int_{\|v\|\le R_n}e^{-\|v\|^2/2}\,dv
=(2\pi)^{3/2}+O(n^{-1/4})+O(e^{-cR_n^2}).
\]

For \(\mathrm{IV}_n(h)\),  \Cref{lem:K-small-frequency} gives \(
\log |K_n(h,u)| = \Re \Log K_n(h,u) \le -\frac{1}{2} u' \Gamma_n(h) u + Cn\norm{u}^3.
\)
For small enough choice of $\rho_0$ and large $n$, we have that $C n\norm{u}^3 \le \frac{1}{4} u'
\Gamma_n(h) u$, since the LHS is $O(n \norm{u}^3)$ and the RHS is $\Theta(n \norm{u}^2)$. Thus
for sufficiently small $\rho_0$ and large $n$, $|K_n(h,u)| \le e^{-\frac{1}{4}u' \Gamma_n(h) u}$. This gives, for large $n$,
\(
|\mathrm{IV}_n(h)|
\le
\int_{\|v\|>R_n}e^{-\frac{1}{4} \|v\|^2}\,dv
\le C e^{-c R_n^2}.
\)
We conclude that, upon verifying that $\det \Gamma_n(h) \approx \det \Gamma_n$ by \Cref{lem:saddlepoint}, \[
	\mathrm{II}_n(h) = \frac{(2\pi)^{3/2} (1+o(1))}{\sqrt{\det \Gamma_n(h)}} = \frac{(2\pi)^{3/2}
	(1+o(1))} {\sqrt{\det \Gamma_n}}.
\]

Turning to \(\mathrm{I}_n(h)\),
independence across $i$ allows
directly computing
\begin{align}
	|K_n(h, u)|^2
	&= \prod_{i=1}^n \abs{(1{-}\pi_{ni}(h)) {+} \pi_{ni}(h) e^{iu'v_{ni}}}^2
	\notag
	\\
	&=
	\prod_v\prod_{i: v_{ni}=v} \abs{1{-}2\pi_{ni}(h) (1{-}\pi_{ni}(h)) (1{-}\cos(u'v))}.
	\label{eq:K_modulus}
\end{align}
Since $\norm{h} \le r_0$ can be chosen to be sufficiently small,
we have $1/2 \ge 2\pi_
{ni} (h)
(1-\pi_{ni}(h)) \ge m > 0$ is bounded away from zero by some $ m = m_{r_0} > 0$. We claim
that at least one $v$ with at least $\eta n$ observations has $\cos(u'v) \le
1-\kappa_\rho$ for some constant $\kappa_\rho > 0$.
Upon verifying this claim in
\Cref{lem:frequent-separated}, we then have that
for this $v$, at least $\eta n$ terms in \eqref{eq:K_modulus} have \(
	[1-2\pi_{ni}(h) (1-\pi_{ni}(h)) (1-\cos(u'v))] < 1-m\kappa_\rho =: r_\rho^2 < 1.
\)
Consequently, \(
	|K_n(h, u)|^2 \le (r_\rho^2)^{\eta n} = e^{-2a_\rho n}
\)
for some $a_\rho > 0$. Finally, we can bound \(
	|\mathrm{I}_n(h)| \le (2\pi)^3 e^{-a_\rho n} = o\pr{
	\frac{1}{\sqrt{\det \Gamma_n}}}.
\)
Combining bounds for $\mathrm{I}_n(h), \mathrm{II}_n(h)$ shows that \[
	\int_{[-\pi, \pi]^3} K_n(h, u) \,du = (1+o(1)) (2\pi)^{3/2} \frac{1}{\sqrt{\det
	\Gamma_n}}.
\]
Further combining with \eqref{eq:expansion_Lambda} verifies the local expansion
\eqref{eq:localexpansion}.

\medskip
It remains to prove the global upper bound.  Taking \(h=0\) in
\eqref{eq:weighted-inversion},
\begin{align*}
P(S_n=z)
&=
(2\pi)^{-3}\int_{[-\pi,\pi]^3}
e^{-iu'(z-\mu_n)}K_n(0,u)\,du,
\\
\implies
\sup_zP(S_n=z)
&\le
(2\pi)^{-3}\int_{[-\pi,\pi]^3}|K_n(0,u)|\,du.
\end{align*}
The same split into \(\|u\|>\rho_0\) and \(\|u\|\le\rho_0\), now with \(h=0\),
gives
\[
	(2\pi)^{-3}\int_{[-\pi,\pi]^3}|K_n(0,u)|\,du
	\le
	\frac{1+o(1)}{(2\pi)^{3/2}\sqrt{\det\Gamma_n}}.
    \qedhere
\]
\end{proof}

\begin{lem}[Saddlepoint and uniform Taylor bounds]\label{lem:saddlepoint}
In the proof of \Cref{prop:lattice}, there is a constant $C$ depending only on $(\lambda,\eta)$ such that, if $\norm{z-\mu_n} \le n^{1/2+\delta}$, there exists
$h = h(z)$ with $\nabla \Lambda_n (h) = z$
and $\norm{h} \le C n^{-1/2 + \delta}$ where \[
	\abs[\bigg]{ \Lambda_n(h) - h'z + \frac{1}{2} (z-\mu_n)' \Gamma_n^{-1} (z-\mu_n)} \le
	C n
	\norm{h}^3.
\]
Moreover, \[
	\abs[\bigg]{\frac{\det \Gamma_n(h)}{\det \Gamma_n} - 1} \le C \norm{h}.
\]
\end{lem}

\begin{proof} It is easy to verify that (i) $\Gamma_n$ has smallest and largest
 eigenvalues within $[cn, Cn]$
 for some $c,C > 0$ that only depend on $(\lambda, \eta)$
 and (ii) the third derivatives of $\Lambda_n$ are uniformly
 bounded by $Cn$, where (i) uses the fact that at least three types of $v_{ni}$ occur at
 least $\eta n$ times, and the three directions form a basis of $\mathbb R^3$ and
 (ii) uses the fact that $\lambda \in (0,1/2)$ and $p_n \in [\lambda, 1-\lambda]$.

A second-order Taylor approximation of $\nabla \Lambda_n$
around $\nabla\Lambda_n
(0) = \mu_n$ shows
that \[
	\nabla \Lambda_n(h) = \mu_n + \Gamma_n h + R_n(h) \quad \norm{R_n(h)} \le Cn
	\norm{h}^2.
\]
Thus if $\nabla \Lambda_n(h) = z$ then $ \Gamma_n h + R_n(h) = z - \mu_n.$ The
 existence of $h$ then boils down to the existence of a fixed point of \( T
 (h) := \Gamma_n^{-1} (z - \mu_n - R_n(h)) \).
 Since $\norm{z- \mu_n} \le n^
 {1/2+\delta}$, we can verify that $\norm{T(h)} \le Cn^{-1/2 + \delta} + A^2 n^
 {-1+2\delta}$ for $\norm{h} \le A n^{-1/2+\delta}$.
 There exists some choices of $A$ so
 that we may apply Brouwer's fixed point theorem\footnote{That is, $T$ maps the ball
 $\br{\norm{h} \le A n^{-1/2+\delta}}$ into itself for sufficiently large $A$ and all
 sufficiently large $n$,
 and $T$ is continuous. } to show that $T (\cdot)$ has a fixed
 point with $\norm{h} \le C n^{-1/2 + \delta}$.

	Now, for this $h$, Taylor expansion of $\Lambda_n(t)$ around $t=0$ shows \[
		\Lambda_n(h) = \underbrace{\Lambda_n(0)}_{0} + (z - (z-\mu_n))'h + \frac12 h'
		\Gamma_n h +
		\xi_n (h) \quad |\xi_n
		(h) |  \le  Cn \norm{h}^3.
	\]
	Rearranging and using the fact that $\Gamma_n^{-1} (z-\mu_n) = h + \Gamma_n^{-1} R_n(h)$,
	we obtain \[
		\Lambda_n(h) - h'z = -\frac{1}{2} (z-\mu_n)'\Gamma_n^{-1} (z-\mu_n) +
		\underbrace{\xi_n (h) +
				\frac{1}{2} R_n(h)' \Gamma_n^{-1} R_n(h)}_{O(n\norm{h}^3)}.
	\]

	For the ``moreover'' part, note that \(
		\norm{\Gamma_n(h) - \Gamma_n}_{op} \le C n \norm{h}
	\).
	Thus, in view of (i), \[
		\frac{\det \Gamma_n(h)}{\det \Gamma_n} = \det(I + \Gamma_n^{-1}\pr{\Gamma_n(h) -
		\Gamma_n}) = 1 + \underbrace{O(\norm{\Gamma_n^{-1}\pr{\Gamma_n(h) -
				\Gamma_n}})}_{O(\norm{h})}.\qedhere
	\]
\end{proof}

\begin{lem}\label{lem:K-small-frequency}
In the proof of \Cref{prop:lattice}, there are constants \(\rho_0,C,b>0\), depending only on
\((\lambda,\eta)\), such that, uniformly over all arrays satisfying the
hypotheses, all \(\|h\|\le r_0\), and all \(\|u\|\le\rho_0\),
\[\Log K_n(h,u) = -\frac12u'\Gamma_n(h)u+\xi_n(h,u),
\qquad
|\xi_n(h,u)|\le Cn\|u\|^3.
\numberthis \label{eq:logKexpansion}
\]
\end{lem}

\begin{proof}
We can explicitly compute that
\[
K_n(h,u)
=
\prod_{i=1}^n
\left\{(1-\pi_{ni}(h))+\pi_{ni}(h)e^{iu'v_{ni}}\right\}
\exp\{-iu'\pi_{ni}(h)v_{ni}\}.
\]
	By Taylor approximation of $\Log((1-\pi) + \pi e^{i q})$ (where $\Log(\cdot)$ is the
	principal branch) around $q = 0$, we have that uniformly over $\pi \in [\lambda',
	1-\lambda']$
	\(
		\abs[\big]{
			\Log\br{(1-\pi) + \pi e^{iq}} - \pr{i\pi q - \frac{1}{2} \pi(1-\pi) q^2
		}}\le C |q|^3
	\)
	for all sufficiently small $|q|$.
	Applying this with $q = u'v_{ni}$ and summing over $i$, after some cancellations,
	yields the desired display
	\eqref{eq:logKexpansion}.
\end{proof}

\begin{lem}\label{lem:frequent-separated}
In the proof of \Cref{prop:lattice},  there exists \(\kappa_\rho>0\), depending only on \(\rho\), such
that the following holds.  Let \(F\subseteq\mathcal V\) contain at least three
directions.  Then, for every \(u\in[-\pi,\pi]^3\) with \(\|u\|\ge\rho\), there is
some \(v\in F\) such that \(\cos(u'v)\le 1-\kappa_\rho\).
\end{lem}

\begin{proof}
It is enough to prove that there is \(c_\rho>0\) such that, for some \(v\in F\),
\(
\dist(u'v,2\pi\mathbb Z)\ge c_\rho.
\)
Indeed, then \(1-\cos(u'v)\ge 1-\cos(c_\rho)=:\kappa_\rho>0\), after replacing
\(c_\rho\) by \(\min(c_\rho,\pi)\).

Suppose no such \(c_\rho\) exists.  Then there are sets \(F_m\subseteq\mathcal V\)
with \(|F_m|\ge3\) and vectors \(u_m\in[-\pi,\pi]^3\), \(\|u_m\|\ge\rho\), such
that
\[
\max_{v\in F_m}\dist(u_m'v,2\pi\mathbb Z)\to0. \numberthis \label{eq:limit}
\]
There are only finitely many subsets of \(\mathcal V\), so along a subsequence
we may take \(F_m=F_*\).  Choose three directions \(v_1,v_2,v_3\in F_*\), and
let \(V_*\) be the \(3\times3\) matrix with these directions as columns.  By
compactness of $[-\pi, \pi]^3$, along another subsequence \(u_m\to u_*\in[-\pi,\pi]^3\),
with
\(\|u_*\|\ge\rho\). \eqref{eq:limit} implies
\(
V_*'u_*\in 2\pi\mathbb Z^3.
\)
Observe that no matter how we choose three vectors from $\mathcal V$,
\((V_*')^{-1}\) has integer entries.  Hence
\(
u_*=(V_*')^{-1}V_*'u_*\in2\pi\mathbb Z^3.
\)
The only point of \(2\pi\mathbb Z^3\cap[-\pi,\pi]^3\) is \(0\), contradicting
\(\|u_*\|\ge\rho\).  This proves the claim.
\end{proof}

\subsection{Application to complete randomization}
\label{sec:application}

\begin{lem}[Complete-randomization likelihood]\label{lem:ftheta-ratio}
Fix \(\lambda\in(0,1/2)\), \(\eta>0\), and \(\delta\in(0,1/6)\), and suppose that
\(p_n:=n_T/n\in[\lambda,1-\lambda]\).  For type counts
\(\theta=(\theta_{at},\theta_{nt},\theta_d,\theta_c) \in \Theta\),
define
\[
s_T(\theta)=\frac{\theta_{at}+\theta_c}{n},
\qquad
s_U(\theta)=\frac{\theta_{at}+\theta_d}{n}.
\]
For \(y=(y_T,y_U)'\), define
\[
t_\theta(y)=
\begin{pmatrix}
y_T-n_Ts_T(\theta)\\
y_U-n_Us_U(\theta)
\end{pmatrix}
= y - \mu_\theta,
\]
\[
D_\theta
=s_T(\theta)(1-s_T(\theta))s_U(\theta)(1-s_U(\theta))
-\left(\frac{\theta_{at}}n-s_T(\theta)s_U(\theta)\right)^2,
\]
and recall \(V_\theta = \tfrac{n-1}{n}\Sigma_\theta\) from the definition of \(\varphi_\theta\), which can be written explicitly as
\[
V_\theta
=
\frac{n_Tn_U}{n}
\begin{pmatrix}
s_T(\theta)(1-s_T(\theta))
& -\left(\frac{\theta_{at}}n-s_T(\theta)s_U(\theta)\right)\\
-\left(\frac{\theta_{at}}n-s_T(\theta)s_U(\theta)\right)
& s_U(\theta)(1-s_U(\theta))
\end{pmatrix}.
\]
Then there is a deterministic sequence
\(\varepsilon_n=\varepsilon_n(\lambda,\eta,\delta)\to0\) such that, uniformly
over \(\theta\in\Theta_{3,n}(\eta)\) and integer pairs \(y\) satisfying
\[
f_\theta(y)>0,\qquad \|t_\theta(y)\|\le n^{1/2+\delta},
\]
we have
\[
f_\theta(y)
=
\frac{1+R_{n,\theta,y}}
{2\pi (n_Tn_U/n)\sqrt{D_\theta}}
\exp\left\{-\frac12t_\theta(y)'V_\theta^{-1}t_\theta(y)\right\},
\qquad
\sup_{\theta,y}|R_{n,\theta,y}|\le\varepsilon_n.
\]
Moreover, uniformly over all \(\theta\in\Theta_{3,n}(\eta)\) and all integer
pairs \(y\),
\[
f_\theta(y)
\le
\frac{1+\varepsilon_n}
{2\pi (n_Tn_U/n)\sqrt{D_\theta}}.
\]
\end{lem}

\begin{proof}
The setup embeds into \Cref{prop:lattice} by taking $X_{ni} = B_i (1,y_i(1),y_i(0))$ for independent Bernoulli($p_n$) variables $B_i$, and writing $N = \sum_i B_i$, $S_1 = \sum_i B_i y_i(1)$, and $S_0 = \sum_i B_i y_i(0)$.
Thus, take $Z = \sum_i X_{ni}$ and for
$
z=(n_T,\ y_T,\ ns_U(\theta)-y_U)',
$ we would like to compute
\begin{equation}\label{eq:ftheta-ratio}
f_\theta(y_T,y_U)
=
P(S_1=y_T,\ S_0=ns_U(\theta)-y_U\mid N=n_T)
=
\frac{P(Z=z)}{P(N=n_T)}.
\end{equation}

The numerator is now exactly in the setup of \Cref{prop:lattice}.
A direct computation gives
\(
(z-E Z)'\Gamma_\theta^{-1}(z-E Z)
=t_\theta(y)'V_\theta^{-1}t_\theta(y).
\)
Finally, Stirling's formula, uniformly for \(p_n\in[\lambda,1-\lambda]\), gives
\[
P(N=n_T)
=
\binom{n}{n_T}p_n^{n_T}(1-p_n)^{n_U}
=
\frac{1+o(1)}{\sqrt{2\pi np_n(1-p_n)}}.
\]
Dividing the local estimate for \(P(Z=z)\) by this last display in
\eqref{eq:ftheta-ratio} proves the asserted local expansion, because
\(\det V_\theta=(n_Tn_U/n)^2D_\theta\).  Dividing the global upper bound for
\(P(Z=z)\) by the same denominator proves the stated upper bound.
\end{proof}

\subsection{Uniform total variation approximation}
\label{sec:tv}

\begin{lem}
\label{lem:central-support}
Fix \(\lambda\in(0,1/2)\), \(\eta>0\), and \(\delta\in(0,1/2)\).  Assume
\(n_T/n\in[\lambda,1-\lambda]\).  For all sufficiently large \(n\), uniformly
over \(\theta\in\Theta_{3,n}(\eta)\), every integer pair \(y\) with
$
\|t_\theta(y)\|\le n^{1/2+\delta}
$
belongs to the support of \(f_\theta\).
\end{lem}

\begin{proof}
We show that every central lattice point can be realized by integer treated
counts in the four type classes.  Let
\[
v_{at}=(1,1,1)',\quad
v_{nt}=(1,0,0)',\quad
v_d=(1,0,1)',\quad
v_c=(1,1,0)'.
\]
For treated counts \(u_j\), the vector \((N,S_1,S_0)'\) equals
\(\sum_j u_jv_j\).  Fix three types whose counts are at least \(\eta n\); call
this set \(A\), and let \(\ell\) be the remaining type.  The vectors
\(\{v_j:j\in A\}\) form a unimodular basis of \(\mathbb Z^3\).

Choose an integer \(u_\ell\in[0,\theta_\ell]\) with
\(|u_\ell-p_n\theta_\ell|\le1\).  Given a central integer pair \(y\), set
\(
z=(n_T,\ y_T,\ ns_U(\theta)-y_U)'.
\)
Since the basis is unimodular, the linear system
\(
\sum_{j\in A}u_jv_j=z-u_\ell v_\ell
\)
has an integer solution.  Relative to the real solution
\((p_n\theta_j)_{j\in A}\), the right-hand side has changed by
\(
z-E Z-(u_\ell-p_n\theta_\ell)v_\ell,
\)
which has norm \(O(n^{1/2+\delta})\).  The inverse of the fixed unimodular basis
matrix has bounded entries, so
\[
u_j=p_n\theta_j+O(n^{1/2+\delta}),\qquad j\in A.
\]
For \(j\in A\), the margins from \(p_n\theta_j\) to the endpoints \(0\) and
\(\theta_j\) are at least \(\min(\lambda,1-\lambda)\eta n\).  Since
\(n^{1/2+\delta}=o(n)\), all these \(u_j\)'s lie in \([0,\theta_j]\) for large
\(n\).  Thus the treated counts are feasible, and \(f_\theta(y)>0\).
\end{proof}

\begin{lem}[Discrete Gaussian normalizer and tails]\label{lem:gaussian-normalizer}
Uniformly over \(\theta\in\Theta_{3,n}(\eta)\),
\(
\sum_{y\in\mathbb Z^2}\varphi_\theta(y)=1+o(1).
\)
Moreover, for every fixed \(\delta>0\),
\[
\sup_{\theta\in\Theta_{3,n}(\eta)}
\sum_{\substack{y\in\mathbb Z^2: \\ \|t_\theta(y)\|>n^{1/2+\delta}}}\gamma_\theta(y)\to0.
\]
\end{lem}

\begin{proof}
The eigenvalues of \(V_\theta\) are uniformly between constant multiples of
\(n\).  For the continuous Gaussian density \(\varphi_\theta\) with mean
\(\mu_\theta\) and covariance \(V_\theta\),
\(
\int_{\mathbb R^2}\varphi_\theta(x)\,dx=1.
\)
A standard bounded-variation Riemann-sum bound gives
\(
\left|
\sum_{y\in\mathbb Z^2}\varphi_\theta(y)
-\int_{\mathbb R^2}\varphi_\theta(x)\,dx
\right|
\le
C\int_{\mathbb R^2}\|\nabla\varphi_\theta(x)\|\,dx.
\)
The right-hand side is \(O(n^{-1/2})\) uniformly, because
\(
\int\|\nabla\varphi_\theta(x)\|\,dx
\le
\{\operatorname{tr}(V_\theta^{-1})\}^{1/2}
=O(n^{-1/2}).
\)
Hence \(\sum_{y\in\mathbb Z^2}\varphi_\theta(y)=1+o(1)\) uniformly.

For the tail bound, the same Riemann-sum comparison and the eigenvalue bounds
give
\(
\sum_{\|t_\theta(y)\|>n^{1/2+\delta}}\varphi_\theta(y)\to0
\)
uniformly over \(\theta\).  Dividing by
\(\sum_{y\in\mathbb Z^2}\varphi_\theta(y)=1+o(1)\) gives the same statement for
\(\gamma_\theta\).
\end{proof}

\begin{prop}[Uniform total variation bound]\label{prop:tv}
For every fixed \(\lambda\in(0,1/2)\) and \(\eta>0\), if
\(n_T/n\in[\lambda,1-\lambda]\), then
\(
\sup_{\theta\in\Theta_{3,n}(\eta)}
\|f_\theta-\gamma_\theta\|_{\mathrm{TV}}
\to0.
\)
\end{prop}

\begin{proof}
Fix any \(\delta\in(0,1/6)\), and let
\(
B_{n,\theta}
=
\{y\in\mathbb Z^2:\|t_\theta(y)\|\le n^{1/2+\delta}\}.
\)
By \Cref{lem:central-support}, every \(y\in B_{n,\theta}\) is in the support of
\(f_\theta\) for all large \(n\).  By \Cref{lem:ftheta-ratio},
\[
f_\theta(y)=(1+R_{n,\theta,y})\varphi_\theta(y),
\qquad
\sup_{\theta,y\in B_{n,\theta}}|R_{n,\theta,y}|\to0.
\]
Also \(\gamma_\theta(y)=\varphi_\theta(y)/\sum_{z\in\mathbb Z^2}\varphi_\theta(z)\),
and \(\sum_{z\in\mathbb Z^2}\varphi_\theta(z)=1+o(1)\) uniformly by
\Cref{lem:gaussian-normalizer}.  Therefore
\[
\sum_{y\in B_{n,\theta}}|f_\theta(y)-\gamma_\theta(y)|
\le
\sup_{y\in B_{n,\theta}}|R_{n,\theta,y}|
\sum_{y\in B_{n,\theta}}\varphi_\theta(y)
+
\left|1-\left(\sum_{z\in\mathbb Z^2}\varphi_\theta(z)\right)^{-1}\right|
\sum_{y\in B_{n,\theta}}\varphi_\theta(y)
=o(1)
\]
uniformly over \(\theta\).

It remains to control the tails.  For \(f_\theta\), the event
\(B_{n,\theta}^c\) is contained in the union of
\[
\left|Y_T-n_Ts_T(\theta)\right|>n^{1/2+\delta}/\sqrt2
\quad\text{and}\quad
\left|Y_U-n_Us_U(\theta)\right|>n^{1/2+\delta}/\sqrt2.
\]
Each term is a hypergeometric tail event, so Hoeffding's inequality for sampling
without replacement gives a bound of order \(e^{-c n^{2\delta}}\), uniformly in
\(\theta\).  Thus
\(
\sup_{\theta\in\Theta_{3,n}(\eta)}f_\theta(B_{n,\theta}^c)\to0.
\)
For \(\gamma_\theta\), the corresponding tail tends to zero by
\Cref{lem:gaussian-normalizer}.  Combining the central and tail bounds gives
the asserted total variation convergence.
\end{proof}

Note that the fixed \(\eta\) condition is essential.
The conclusion is uniform on \(\Theta_{3,n}(\eta)\) for fixed \(\eta>0\).  If
``at least three categories'' only means nonzero, with no fixed lower bound on
the third-largest cell count, the uniform claim generally fails: \(D_\theta\)
can approach zero, destroying the order-\(n\) nonsingularity and uniform
aperiodicity constants used above.

\subsection{Le Cam convergence}
\label{sec:lecamconv}

\begin{proof}[Proof of \Cref{prop:lecam}]
Since the experiments \(\mathcal E_n\) and \(\mathcal G_n\) have the same sample space, the identity Markov kernel
maps each observed lattice point to itself.  Therefore both deficiencies are
bounded by the same uniform total variation distance:
\[
\delta(\mathcal E_n,\mathcal G_n)
\le
\sup_{\theta\in\Theta_{3,n}(\eta)}
\|f_\theta-\gamma_\theta\|_{\mathrm{TV}},
\qquad
\delta(\mathcal G_n,\mathcal E_n)
\le
\sup_{\theta\in\Theta_{3,n}(\eta)}
\|f_\theta-\gamma_\theta\|_{\mathrm{TV}}.
\]
The right-hand side tends to zero by \Cref{prop:tv}.  Hence
\[
\Delta(\mathcal E_n,\mathcal G_n)
=
\max\pr{\delta(\mathcal E_n,\mathcal G_n),
\delta(\mathcal G_n,\mathcal E_n)}
\to0.
\qedhere
\]
\end{proof}

\subsection{MLE convergence}
\label{sec:mleconv}

We now use the local likelihood approximation above to prove \Cref{prop:mlefullconvergence}, the convergence of
the design-based MLE; throughout this subsection, write \(\hat\theta = \hat{\theta}^{\text{MLE}}\).
Let \(\lfloor r\rceil\) denote any nearest integer to \(r\).  Define
\[
x(\theta)=\frac{\theta_{at}}n,\qquad
D(s_T,s_U,x)
=s_T(1-s_T)s_U(1-s_U)-(x-s_Ts_U)^2.
\]
Also write
\[
L(s_T,s_U)=\max(0,s_T+s_U-1),
\qquad
U(s_T,s_U)=\min(s_T,s_U).
\]

Let the true finite populations satisfy
\[
\frac{\theta_n}{n}\to
p=(p_{at},p_{nt},p_d,p_c)\in[0,1]^4,
\qquad
\frac{n_T}{n}\in[\lambda,1-\lambda]
\]
for all large \(n\), where \(\lambda\in(0,1/2)\); the individual limiting shares may lie on the boundary, as the true \(\theta_n\) enters the argument below only through the convergence of its margins \(s_T(\theta_n)\) and \(s_U(\theta_n)\).  Define
$
q_T=p_{at}+p_c$ and $q_U=p_{at}+p_d,
$
and assume
\[
q_T,q_U\in(0,1)\setminus\{1/2\},\qquad
q_T\ne q_U,\qquad q_T+q_U\ne1. \tag{A}\label{ass:mle-A}
\]
Parametrize the line segment \(\mathcal{P}(q_T,q_U)\) from \Cref{prop:mlefullconvergence} by \(x=p'_{at}\):
$
p'(x)=
\left(x,\ 1-q_T-q_U+x,\ q_U-x,\ q_T-x\right),
$
with
$x \in [L^*, U^*]$
where
$
L^*=\max(0,q_T+q_U-1)$ and
$U^*=\min(q_T,q_U).$

\begin{proof}[Proof of \Cref{prop:mlefullconvergence}]
Without essential loss, assume that
$
D(q_T,q_U,L^*)<D(q_T,q_U,U^*)
$ and we will show $x(\hat\theta) \pto L^*$.
Under \eqref{ass:mle-A}, \(L^*<U^*\), \(D(q_T,q_U,x)>0\) on
\([L^*,U^*]\), and
$
D(q_T,q_U,L^*) \neq D(q_T,q_U,U^*)$, where the latter holds since \(x \mapsto D(q_T,q_U,x)\) is a downward parabola with vertex at \(x = q_Tq_U\), so that equal endpoint values would make \(q_Tq_U\) the midpoint of \([L^*,U^*]\) and force \(q_T = \tfrac{1}{2}\) or \(q_U = \tfrac{1}{2}\). Observe that with a direct computation the endpoint $
\br{L^*, U^*}$ with the smaller
value of $D(q_T, q_U, \cdot)$ is the unique maximizer of $x\mapsto \max_j p'_j(x)$ over
$x \in [L^*, U^*]$.
Since $L^*$ maximizes $\max_j p'_j(x)$ in this
case, showing $x(\hat\theta) \pto L^*$ proves the proposition.
The other case with $U^*$ is analogous.

Choose \(\zeta>0\) small enough that
$
L^*+\zeta<U^*-\zeta
$
and
\[
D(q_T,q_U,L^*+\zeta)
<
\inf_{x\in[L^*+2\zeta,U^*]}D(q_T,q_U,x). \numberthis \label{eq:bad-set-d-gap}
\]
This is possible because \(q_Tq_U\in(L^*,U^*)\), \(D\) is increasing to the left
of \(q_Tq_U\), and \(D(q_T,q_U,L^*)<D(q_T,q_U,U^*)\). For \(\alpha>0\), define the
following set
\[
\widetilde\Theta^{\mathrm{bad}}_n(\alpha,\zeta)
=
\left\{
\theta\in\Theta:
|s_T(\theta)-q_T|+|s_U(\theta)-q_U|\le\alpha,\
x(\theta)\in[L^*+2\zeta,\ U(s_T(\theta),s_U(\theta))]
\right\}.
\]
{Fix \(\alpha>0\) sufficiently small.  It is easy to check that there is
\(\eta>0\), independent of \(n\), such that
$
\widetilde\Theta^{\mathrm{bad}}_n(\alpha,\zeta)
\subseteq\Theta_{3,n}(\eta)
$
for all sufficiently large \(n\).

We claim that
\(
P\{\hat\theta\in\widetilde\Theta^{\mathrm{bad}}_n(\alpha,\zeta)\}\to0.
\)
Define the event
\[
\mathcal A_n
=
\left\{
\left|\frac{Y_T}{n_T}-s_T(\theta_n)\right|
+
\left|\frac{Y_U}{n_U}-s_U(\theta_n)\right|
\le n^{-1/4}
\right\}.
\]
Because \(n_T/n\in[\lambda,1-\lambda]\), Hoeffding's inequality for sampling
without replacement and a union bound give \(P(\mathcal A_n^c)=o(1)\).  Moreover,
uniformly on \(\mathcal A_n\),
\[
\left|\frac{Y_T}{n_T}-q_T\right|
+
\left|\frac{Y_U}{n_U}-q_U\right|
\le
n^{-1/4}
+|s_T(\theta_n)-q_T|
+|s_U(\theta_n)-q_U|
=o(1).
\]

For the observed \(Y=(Y_T,Y_U)\), define a competitor
\(\theta^L=\theta^L(Y)\) by
\[
\theta^L_{at}=\left\lfloor n(L^*+\zeta)\right\rceil,\qquad
\theta^L_{at}+\theta^L_c=\left\lfloor nY_T/n_T\right\rceil,\qquad
\theta^L_{at}+\theta^L_d=\left\lfloor nY_U/n_U\right\rceil,
\]
with \(\theta^L_{nt}=n-\theta^L_{at}-\theta^L_c-\theta^L_d\).  Componentwise,
\[
\frac{\theta^L}{n}
=
\left(
L^*+\zeta,\,
1-\frac{Y_T}{n_T}-\frac{Y_U}{n_U}+L^*+\zeta,\,
\frac{Y_U}{n_U}-L^*-\zeta,\,
\frac{Y_T}{n_T}-L^*-\zeta
\right)
+O(n^{-1}).
\]
Since \(L^*+\zeta\in(L^*,U^*)\), the vector on the right converges uniformly
on \(\mathcal A_n\) to \(p'(L^*+\zeta) = \left(
L^*+\zeta,\;
1-q_T-q_U+L^*+\zeta,\;
q_U-L^*-\zeta,\;
q_T-L^*-\zeta
\right)\), all four coordinates of which are
strictly positive.  Hence there is an \(\eta'>0\), independent of \(n\), such
that \(\theta^L\in\Theta_{3,n}(\eta')\) for all sufficiently large \(n\) on
\(\mathcal A_n\).  By construction, 
\[
\|t_{\theta^L}(Y_T,Y_U)\|=O(1)
\]
and, uniformly on \(\mathcal A_n\),
\[
\begin{aligned}
D_{\theta^L}
&=
D\left(\frac{Y_T}{n_T},\frac{Y_U}{n_U},L^*+\zeta\right)+O(n^{-1})=
D(q_T,q_U,L^*+\zeta)+o(1).
\end{aligned}
\]
Therefore \Cref{lem:central-support} applies.  In addition,
\(V_{\theta^L}^{-1}=O(n^{-1})\) uniformly, so the exponential term in
\Cref{lem:ftheta-ratio} is \(1+o(1)\).  It follows that, uniformly on
\(\mathcal A_n\),
\[
f_{\theta^L}(Y_T,Y_U)
=
\frac{1+o(1)}
{2\pi(n_Tn_U/n)\sqrt{D(q_T,q_U,L^*+\zeta)}}. \numberthis
\label{eq:lower-comp-unified}
\]}

On the other hand, by \eqref{eq:bad-set-d-gap} and continuity, after reducing
\(\alpha\) if necessary,\footnote{Let
\(
g_\zeta
=
\inf_{x\in[L^*+2\zeta,U^*]}D(q_T,q_U,x)
-D(q_T,q_U,L^*+\zeta)>0 .
\)
For \(\alpha>0\), set
\(
K_\alpha
=
\{(s_T,s_U,x): |s_T-q_T|+|s_U-q_U|\le\alpha,\
x\in[L^*+2\zeta,U(s_T,s_U)]\}.
\)
Then \(K_\alpha\downarrow K_0=\{(q_T,q_U,x):x\in[L^*+2\zeta,U^*]\}\) as
\(\alpha\downarrow0\).  Hence, by continuity of \(D\),
\(
\inf_{(s_T,s_U,x)\in K_\alpha}D(s_T,s_U,x)
\ge D(q_T,q_U,L^*+\zeta)+g_\zeta/2
\)
for all sufficiently small \(\alpha\).}
\(
\inf_{\theta\in\widetilde\Theta^{\mathrm{bad}}_n(\alpha,\zeta)}D_\theta
>
D(q_T,q_U,L^*+\zeta) + c_\zeta
\)
for some $c_\zeta > 0$ not dependent on $n$.  The global upper bound in
\Cref{lem:ftheta-ratio} therefore
implies
\[
\sup_{\theta\in\widetilde\Theta^{\mathrm{bad}}_n(\alpha,\zeta)}
f_\theta(Y_T,Y_U)
\le
\frac{1+o(1)}
{2\pi(n_Tn_U/n)
\sqrt{\inf_{\theta\in\widetilde\Theta^{\mathrm{bad}}_n(\alpha,\zeta)}D_\theta}}.
\]
Combining this display with \eqref{eq:lower-comp-unified},
\[
\sup_{\theta\in\widetilde\Theta^{\mathrm{bad}}_n(\alpha,\zeta)}
\frac{f_\theta(Y_T,Y_U)}{f_{\theta^L}(Y_T,Y_U)}
\le
(1+o(1))
\left\{
\frac{D(q_T,q_U,L^*+\zeta)}
{\inf_{\theta\in\widetilde\Theta^{\mathrm{bad}}_n(\alpha,\zeta)}D_\theta}
\right\}^{1/2}
<1
\]
{for all sufficiently large \(n\) on \(\mathcal A_n\).  Thus, for
all sufficiently large \(n\),
\[
\{\hat\theta\in\widetilde\Theta^{\mathrm{bad}}_n(\alpha,\zeta)\}
\cap\mathcal A_n
=\varnothing,
\]
because every parameter in $\widetilde\Theta^{\mathrm{bad}}_n(\alpha,\zeta)$ has strictly
smaller likelihood than
the feasible competitor \(\theta^L\).  Consequently,
\[
\begin{aligned}
P\{\hat\theta\in\widetilde\Theta^{\mathrm{bad}}_n(\alpha,\zeta)\}
&\le
P\bigl(
\{\hat\theta\in\widetilde\Theta^{\mathrm{bad}}_n(\alpha,\zeta)\}
\cap\mathcal A_n
\bigr)
+P(\mathcal A_n^c) =o(1).
\end{aligned}
\]

Finally, feasibility ensures
\(x(\hat\theta)\le U(s_T(\hat\theta),s_U(\hat\theta))\).  Hence
\[
\begin{aligned}
P\{x(\hat\theta)\ge L^*+2\zeta\}
&\le
P\{\hat\theta\in\widetilde\Theta^{\mathrm{bad}}_n(\alpha,\zeta)\} + P\bigl\{
|s_T(\hat\theta)-q_T|+|s_U(\hat\theta)-q_U|>\alpha
\bigr\}
\to0
\end{aligned}
\]
by \Cref{lem:mle-margin-consistency}.}
Feasibility gives \(x(\hat\theta)\ge L(s_T(\hat\theta),s_U(\hat\theta))\), and
\Cref{lem:mle-margin-consistency} implies
$L(s_T(\hat\theta),s_U(\hat\theta))\pto L^*.
$ Since $\zeta$ is arbitrary,
 \(x(\hat\theta)\pto L^*\).

Together with \Cref{lem:mle-margin-consistency}, this yields
\[
\tfrac{1}{n}\hat\theta
=
\pr{x(\hat\theta),\; 1 - s_T(\hat\theta) - s_U(\hat\theta) + x(\hat\theta),\; s_U(\hat\theta) - x(\hat\theta),\; s_T(\hat\theta) - x(\hat\theta)}
\pto
p'(L^*),
\]
which depends on \(p\) only through \((q_T,q_U)\).
Finally, under \eqref{ass:mle-A}, exactly one coordinate of the limit is zero: at \(L^*\), either \(p'_{at}(0)=0\) when \(q_T+q_U<1\) or \(p'_{nt}(q_T{+}q_U{-}1)=0\) when \(q_T+q_U>1\), with the remaining three coordinates strictly positive; analogously at \(U^*\), \(p'_d(U^*)=0\) when \(q_U<q_T\) and \(p'_c(U^*)=0\) when \(q_T<q_U\).
\end{proof}

\begin{lem}\label{lem:mle-margin-consistency}
Under \eqref{ass:mle-A},
$
s_T(\hat\theta)\pto q_T$ and $
s_U(\hat\theta)\pto q_U.
$
\end{lem}

\begin{proof}[Proof of \Cref{lem:mle-margin-consistency}]
Fix \(\varepsilon>0\) and choose \(\gamma<\varepsilon/2\).  Let
\(
\mathcal A_{n,\gamma}
=
\left\{
\left|\frac{Y_T}{n_T}-q_T\right|
+
\left|\frac{Y_U}{n_U}-q_U\right|
\le \gamma
\right\}.
\)
Hoeffding's inequality for sampling without replacement gives
\(P(\mathcal A_{n,\gamma})\to1\).

On \(\mathcal A_{n,\gamma}\), set
\[
M_T=\left\lfloor nY_T/n_T\right\rceil,\qquad
M_U=\left\lfloor nY_U/n_U\right\rceil
\]
and define
\[
\tilde\theta_{at}=\left\lfloor M_TM_U/n\right\rceil,\qquad
\tilde\theta_c=M_T-\tilde\theta_{at},\qquad
\tilde\theta_d=M_U-\tilde\theta_{at},
\qquad
\tilde\theta_{nt}=n-M_T-M_U+\tilde\theta_{at}.
\]
For small enough \(\gamma\), \(\tilde\theta\in\Theta_{3,n}(\eta)\) for some
\(\eta>0\), \(\|t_{\tilde\theta}(Y_T,Y_U)\|=O(1)\), and
\(D_{\tilde\theta}\) is bounded away from zero.  Hence
\Cref{lem:central-support} and \Cref{lem:ftheta-ratio} imply that, for some
\(C>0\),
\(
f_{\tilde\theta}(Y_T,Y_U)\ge Cn^{-1}
\)
for all large \(n\) on \(\mathcal A_{n,\gamma}\).  Therefore
\(f_{\hat\theta}(Y_T,Y_U)\ge Cn^{-1}\) on the same event.

If \(|s_T(\theta)-q_T|>\varepsilon\) and \(\mathcal A_{n,\gamma}\) holds, then
\(
\left|Y_T-n_Ts_T(\theta)\right|\ge n_T(\varepsilon-\gamma)\ge n_T\varepsilon/2.
\)
Thus Hoeffding's inequality gives, uniformly over such \(\theta\),
\[
f_\theta(Y_T,Y_U)
\le
P_\theta\left\{|Y_T-n_Ts_T(\theta)|\ge n_T\varepsilon/2\right\}
\le 2e^{-cn}.
\]
For large \(n\), \(2e^{-cn}<Cn^{-1}\), so no such \(\theta\) can maximize the
likelihood on \(\mathcal A_{n,\gamma}\).  The same argument applies to
\(s_U\), proving the claim.
\end{proof}

\end{document}